\documentclass[12pt,a4paper]{amsart} 
\usepackage{amsaddr}
\usepackage[T1]{fontenc}
\usepackage[margin=1.5cm]{geometry}
\usepackage[pdfusetitle]{hyperref}
\usepackage[normalem]{ulem}
\usepackage{wasysym}

\usepackage{amscd,amsmath,amssymb,amsfonts,xspace,mathrsfs,amsthm}
\usepackage{color, mathtools}
\usepackage[dvipsnames]{xcolor}

\usepackage{url}

\usepackage{latexsym}
\usepackage{graphicx}
\usepackage{dsfont}
\usepackage{longtable}
\usepackage{enumitem}

\usepackage[latin1]{inputenc}
\usepackage{multirow}

\newtheorem{lemma}{Lemma}

\newtheorem{remark}{Remark}

\usepackage{tikz-cd}
\usetikzlibrary{bending}
\usepackage{tikz,textcomp}
\usetikzlibrary{decorations.pathreplacing,calc,fadings,fit,shapes,arrows,positioning,chains,matrix}

\numberwithin{equation}{section}

\def\bea{\begin{eqnarray}}
\def\eea{\end{eqnarray}}
\def\be{\begin{equation}}
\def\ee{\end{equation}}
\def\ba{\begin{align}}
\def\ea{\end{align}}
\def\bse{\begin{subequations}}
\def\ese{\end{subequations}}

\newcommand{\nn}{\nonumber}

\def\PE{{\rm PE }}

\DeclareMathOperator{\Td}{Td}
\DeclareMathOperator{\ch}{ch}
\DeclareMathOperator{\rk}{rk}

\DeclareMathOperator{\Coh}{Coh}

\DeclareMathOperator{\Stab}{Stab}
\DeclareMathOperator{\Pic}{Pic}

\def\({\left(}
\def\){\right)}
\def\[{\left[}
\def\]{\right]}
\def\<{\left\langle}
\def\>{\right\rangle}
\def\hf{{1\over 2}}

\renewcommand\v{\mathsf v}

\newcommand{\eps}{\epsilon}

\newcommand{\I}{\mathrm{i}}

\newcommand{\cA}{\mathcal{A}}

\newcommand{\cC}{\mathcal{C}}
\newcommand{\cD}{\mathcal{D}}
\newcommand{\cE}{\mathcal{E}}

\newcommand{\cH}{\mathcal{H}}

\newcommand{\cO}{\mathcal{O}}

\newcommand{\Z}{{\mathbb Z}}

\newcommand{\IR}{\mathds{R}}

\newcommand{\IZ}{\mathds{Z}}
\newcommand{\IQ}{\mathds{Q}}

\newcommand{\IH}{\mathds{H}}

\newcommand{\IP}{\mathds{P}}

\def\scM{\mathscr{M}}

\def\hq{\hat q}

\def\CY{\mathfrak{Y}}

\def\bOm{\overline{\Omega}}
\def\bOmPi{\lefteqn{\overline{\phantom{\Omega}}}\Omega^\Pi}

\def\tOm{\widetilde{\Omega}}

\def\hint{h^{\rm (int)}}

\def\han{h^{\rm (an)}}
\def\hh{h^{(0)}}

\def\vths#1{\vartheta^{(#1)}}

\def\Gi#1{G^{(#1)}}

\def\mm{\ell_0}

\def\rmz{{\rm z}}

\def\gmax{g_{\rm max}}

\def\bOmH{\bOm_H}
\def\chiOD{\chi_{\cD}}
\def\DDb{D6-$\overline{\rm D6}$\ }
\def\tvths#1{\tilde\vartheta^{(#1)}}

\newcommand{\q}{\mbox{q}}

\newcommand\PT{\operatorname{PT}}
\newcommand\DT{\operatorname{DT}}

\newcommand\dt{\operatorname{DT}}

\def\GV{{\rm GV}}
\newcommand{\GVg}[2][Q]{{\GV}^{(#2)}_{#1}}

\newcommand\beq[1]{\begin{equation}\label{#1}}
\newcommand\eeq{\end{equation}}
\newcommand\beqa{\begin{eqnarray*}}
\newcommand\eeqa{\end{eqnarray*}}

\makeatletter
\newtheorem*{rep@theorem}{\rep@title}
\newcommand{\newreptheorem}[2]{%
\newenvironment{rep#1}[1]{%
 \def\rep@title{#2 \ref{##1}}%
 \begin{rep@theorem}}%
 {\end{rep@theorem}}}
\makeatother

\newtheorem{Thm}{Theorem}
\newreptheorem{Thm}{Theorem}
\newtheorem{Thm*}{Theorem}
\newtheorem{Prop}[Thm]{Proposition}

\newtheorem{Lem}[Thm]{Lemma}

\newreptheorem{Cor}{Corollary}

\newtheorem{thm-int}{Theorem}
\theoremstyle{definition}
\newtheorem{Def-s}[Thm]{Definition}

\newenvironment{enumerate*}{\begin{enumerate}}{\end{enumerate}}

\title{
Quantum geometry and mock modularity
}

\author{Sergei Alexandrov}
\address{Laboratoire Charles Coulomb (L2C), Universit\'e de Montpellier,
CNRS, \\ F-34095, Montpellier, France}
\email{sergey.alexandrov@umontpellier.fr}

\author{Soheyla Feyzbakhsh}
\address{ Department of Mathematics, Imperial College, London SW7 2AZ, United Kingdom}
\email{s.feyzbakhsh@imperial.ac.uk}

\author{Albrecht Klemm}
\address{Bethe Center for Theoretical Physics and 
Hausdorff Center for Mathematics,\\ Universit\"at Bonn, D-53115, Germany}
\email{aoklemm@th.physik.uni-bonn.de}

\author{Boris Pioline}
\address{Sorbonne Universit\'e, CNRS, 
Laboratoire de Physique Th\'eorique et Hautes Energies, \\
Campus Pierre et Marie Curie, 4 place Jussieu, F-75005, Paris, France}
\email{pioline@lpthe.jussieu.fr}

\begin{document}
\setlength{\parskip}{0.2cm}

\begin{abstract}
In previous work, we used new mathematical relations between 
Gopakumar-Vafa (GV) invariants and rank 0 Donaldson-Thomas (DT) invariants 
to determine the first few terms in the generating series of Abelian D4-D2-D0 indices for a class of  compact one-parameter Calabi-Yau threefolds. This allowed us to obtain striking checks of S-duality, namely the prediction that these series should be vector-valued weakly holomorphic modular forms under $SL(2,\IZ)$.
In this work, we extend this analysis to the case of D4-D2-D0 indices with two units of
D4-brane charge, where S-duality instead predicts that the corresponding generating series should be mock modular with a specific shadow. For the  degree 10 hypersurface in weighted projective space $\IP_{5,2,1,1,1}$, and the degree 8 hypersurface in $\IP_{4,1,1,1,1}$, where GV invariants can be computed to sufficiently high genus, we find that the first few terms indeed match a unique mock modular form with the required
properties, which we determine explicitly. Turning the argument around, we obtain new boundary conditions on the holomorphic ambiguity of the topological string amplitude, which in principle allow to determine it completely up to genus 95 and 112, respectively, i.e. almost twice the maximal genus obtainable using gap and ordinary Castelnuovo vanishing conditions. 
\end{abstract}

\maketitle

\newpage
\tableofcontents

\setlength{\parskip}{0.2cm}

\section{Introduction}

One of the most striking predictions of the duality between type IIA strings and M-theory is that 
a large class of BPS black holes in type IIA strings compactified on a Calabi-Yau threefold $\CY$ (namely those with vanishing D6-brane charge)
are BPS black strings in disguise, obtained by wrapping an M5-brane on a suitable divisor $\cD\subset \CY$~\cite{Maldacena:1997de}. Consequently, the generating series of BPS indices counting D4-D2-D0 black hole microstates are identified with the elliptic genus of the black string 
superconformal field theory, and should therefore enjoy modular invariance~\cite{Maldacena:1997de,deBoer:2006vg,Gaiotto:2006wm,
Gaiotto:2007cd,Manschot:2007ha}. From the mathematical viewpoint
interpreting the physical BPS indices as the Donaldson-Thomas (DT) invariants $\Omega_\sigma(\gamma)$
(where $\sigma$ denotes a stability condition on the derived category of coherent sheaves 
$D^b\Coh( \CY)$, determined by the complexified K\"ahler structure on  
$\CY$), the origin of these modular properties is still largely mysterious, except in the special case when $\cD$ is a vertical divisor in an 
elliptic  \cite{Klemm:2012sx,Oberdieck:2016nvt} or K3 fibration \cite{Bouchard:2016lfg}. Moreover, when the divisor $\cD$ is reducible, a more careful analysis \cite{Alexandrov:2016tnf,Alexandrov:2017qhn,Alexandrov:2018lgp}  reveals that the generating series should in fact be a (vector valued, higher depth) mock modular form, with a specific modular anomaly determined recursively from the generating series attached to the components $\cD_i$ into which $\cD=\sum_i \cD_i$ decomposes (see \cite{Manschot:2010xp,Alim:2010cf,Dabholkar:2012nd,Cheng:2017dlj} for related work).

Since the appropriate space of modular forms is finite-dimensional, the full generating series of D4-D2-D0 indices is in principle 
fixed by a finite number of Fourier coefficients, for example the polar coefficients associated to negative powers of $\q$, the fugacity conjugate to D0-brane charge. This strategy  
was applied long ago to determine  D4-D2-D0 indices with one unit 
of D4-brane charge (also called Abelian) for the quintic in 
$\IP^4$, as well as for a few other CY threefolds  with one K\"ahler modulus (i.e. $b_2(\CY)=1$) in \cite{Gaiotto:2006wm,Gaiotto:2007cd,Collinucci:2008ht,VanHerck:2009ww}. In these works, a candidate for the generating series of D4-D2-D0 invariants with $\cD=H$ (where $H$ is the generator of the one-dimensional lattice $H_4(\CY,\IZ)=\IZ H$) 
was found from heuristic computations of the lowest coefficients in the $q$-expansion. 
This approach was revisited in \cite{Alexandrov:2022pgd}, where it was pointed out that similar physical arguments for other
one-parameter CY threefolds lead to coefficients incompatible with modularity.
Instead, in our last work \cite{Alexandrov:2023zjb}, we exploited recent 
mathematical progress in Donaldson-Thomas theory 
to compute rigorously the polar coefficients 
as well as a large number of non-polar coefficients, and found a precise match with a unique 
vector-valued modular form for most hypergeometric CY threefolds, providing striking evidence for S-duality in string theory (and correcting, in a few instances, the naive physical predictions from~\cite{Gaiotto:2006wm,Gaiotto:2007cd,Collinucci:2008ht,Alexandrov:2022pgd}). 

Our goal in this work is to extend the analysis of \cite{Alexandrov:2023zjb} to the case of D4-D2-D0 invariants 
with 2 units of D4-brane charge (i.e. $\cD=2H$), and identify the mock modular forms representing their generating series. 
The first step in this direction was carried out in  \cite{Alexandrov:2022pgd} where the generating series $h_{2,\mu}$
of these invariants was decomposed into the sum of an ordinary, vector-valued weak modular form $\hh_{2,\mu}$ 
and a specific mock modular form $\han_{2,\mu}$, constructed out of 
the standard generating series 
of Hurwitz class numbers \cite{Zagier:1975} and having the same modular anomaly as $h_{2,\mu}$,
\be
h_{2,\mu}=\hh_{2,\mu}+\han_{2,\mu}.
\label{hhh}
\ee
This reduced the problem to finding the ordinary weak modular form $\hh_{2,\mu}$, and thus again to determining
its polar coefficients. In this work we shall achieve this goal
for some of the simplest one-modulus CY threefolds,
by generalizing the relations between various topological invariants derived in \cite{Alexandrov:2023zjb}.
The CY threefolds that we will consider are
the degree 10 hypersurface in weighted projective space $\IP_{5,2,1,1,1}$, which we denote by $X_{10}$, 
and the degree 8 hypersurface in $\IP_{4,1,1,1,1}$, denoted by $X_8$.

The key idea behind the mathematical results~\cite{Toda:2011aa, Feyzbakhsh:2020wvm,Feyzbakhsh:2021rcv,Feyzbakhsh:2021nds,Feyzbakhsh:2022ydn}, on which 
our previous work \cite{Alexandrov:2023zjb} is based, is to study wall-crossing in a complex 
one-dimensional family of weak stability conditions on the bounded derived category of coherent sheaves 
$D^b\Coh\CY$, which is the correct mathematical description of BPS states 
in type IIA string theory compactified on $\CY$. This family of weak stability conditions agrees with the physical stability condition in the infinite volume limit (i.e. when the K\"ahler modulus $\rmz \to\I\infty$), but has a much simpler chamber structure. Moreover, for sufficiently `small' charge $\gamma$ (see \eqref{condsmall} below for the precise smallness criterium), it follows from the Bayer-Macr\`\i-Toda (BMT) inequality conjectured in  \cite{bayer2011bridgeland} that there exists a chamber $\sigma(\gamma)$ where the Donaldson-Thomas invariant $\Omega_{\sigma(\gamma)}(\gamma)$ vanishes.
The BMT inequality has been proven for the quintic in $\IP^4$ \cite{li2019stability} and a handful of other models, but it is widely believed to hold for any CY threefold, as it is a very natural step in the construction of Bridgeland stability conditions.\footnote{While the
physical origin of the BMT inequality is still mysterious, in \S\ref{sec_weak} we offer a physical interpretation of one of its consequences in terms of the entropy of single-centered black holes. See also \cite{Halder:2023kza} for a similar interpretation in the context of five-dimensional black holes.}
Starting from this vanishing chamber, one can in principle compute 
$\Omega_{\sigma}(\gamma)$ for any (weak) stability condition using the wall-crossing formula
of \cite{ks,Joyce:2008pc}. 

Specifically, by considering wall-crossing for charge
$\gamma$  with $-1$ unit of D6-brane charge and D2-D0 charge $(Q,m)$ sufficiently close to the Castelnuovo bound (see \eqref{CastPT} below), 
one can express the index at large volume as a sum of wall-crossing contributions where the anti-D6-brane emits D4-D2-D0 bound states with charges $\gamma_i$, which are the states that we aim to count. The index $\Omega(\gamma)$ at large volume is otherwise known as the Pandharipande-Thomas invariant $\PT(Q,m)$, and is computable in terms of the Gopakumar-Vafa invariants $\GVg{g}$, which determine the topological string partition function at large volume. The latter are in turn computable by integrating
the holomorphic anomaly equations, so long as sufficiently many conditions are available at each genus to fix the holomorphic ambiguity. In \cite{Alexandrov:2023zjb}, we restricted to cases where 
the charge $\gamma$ is so close to the Castelnuovo bound that the only D4-D2-D0 charges $\gamma_i$ contributing to wall-crossing had unit D4-brane charge. This provides an efficient way of computing the 
Abelian D4-D2-D0 invariants for small enough (but not necessarily polar) D0-brane charge. Conversely, having computed sufficiently many of these Abelian D4-D2-D0 invariants so as to uniquely pin down the modular generating series, one can input the resulting infinite series of invariants to compute PT invariants at larger charge $m$, or GV invariants at higher genus $g$, and hence obtain new boundary conditions for the direct integration method. 

\begin{table} [t]
\begin{centering}
$$
\begin{array}{|l|r|r|r|r|r|r|r|r|r|r|r|r|r|r|r|r|}
\hline \CY  & \chi& \kappa  &c_{2} 
& n_1  & C_1 & n_2  & C_2  
& g_{\rm integ} & g_{\rm mod}^{(1)} & g_{\rm mod}^{(2)}  & g_{\rm avail}\\   \hline
X_5(1^5)   
& -200   &5   &  50 
& 7 & 0 & 36 & 1 
& 53 & 69 &  80  & 64
\\
X_6(1^4,2)  
& -204& 3 & 42
& 4 & 0 & 19 & 1 
& 48 & 66 & 84 & 48
\\
X_8(1^4,4)  
&-296  &2 & 44  
&  4 & 0  & 14 & 1 
& 60 & 84 & 112 & 64
\\
X_{10}(1^3,2,5)  
&  -288&  1 & 34
& 2   & 0 & 7 & 0 
& 50 & 70 & 95 & 67
\\
\hline
\end{array}
\vspace{0.2cm}
$$
\caption{Relevant data for the CY hypersurfaces  of degree 5, 6, 8 and 10.  
The columns $n_1,C_1$ and $n_2,C_2$ indicate the number of polar terms and modular constraints
on the generating series of D4-D2-D0 invariants with $r=1$ and $r=2$, respectively, taken from \cite{Alexandrov:2022pgd}.
The  column $g_{\rm integ}$  indicates the maximal genus for which GV invariants $\GV_Q^{(g)}$ can be determined by standard
direct integration,  using only the usual regularity conditions and the known expression for GV invariants saturating the bound
$g\leq g_{\rm max}(Q)$ for $Q=0 \mod \kappa$. The columns $g_{\rm mod}^{(r)}$ with $r=1$ and 2 indicate the reachable genus assuming that D4-D2-D0 indices with D4-brane charge up to $r$ are known.
The column $g_{\rm avail}$ indices the genus up to which 
GV invariants have been effectively computed at this stage.
\label{table1}}
\end{centering}
\end{table}

In the present work, we shall weaken the restrictions on $\gamma$ such that D4-D2-D0 charges $\gamma_i$ with two units of D4-brane charge can also contribute, giving us access to D4-D2-D0 indices with $\cD=2H$. In this case, the primitive wall-crossing formula of \cite{Denef:2007vg} is no longer sufficient, instead one needs the  general wall-crossing formula of \cite{ks,Joyce:2008pc}. Moreover, one requires GV invariants of even higher genus than in the Abelian case in order to compute all polar terms. For the two models $X_8$ and $X_{10}$ considered in this work, it turns out that our current knowledge of GV invariants enables us to compute a sufficient number of terms so that we can identify a candidate mock modular form uniquely and obtain several very non-trivial checks. Conversely, by inputting the infinite series
of  invariants with two units of D4-brane charge, we obtain an infinite number of new conditions on the topological string partition function, which in principle allow us to push   the maximal genus attainable by the direct integration method  by a factor of roughly $1.8$, see Table 
\ref{table1} and Figure \ref{figXGV}. Unfortunately, our current knowledge of GV invariants does not allow us to 
carry out this program for other one-parameter CY threefolds 
at this point.

The remainder of this work is organized as follows. In \S\ref{sec_review}, we recall the notations from \cite{Alexandrov:2023zjb} and the main
relations satisfied (or expected to be satisfied) by DT, PT and GV invariants, 
specializing to the case of CY threefolds
with Picard rank 1 (hence a single K\"ahler parameter). 
In passing, in \S\ref{sec_weak} we give a physical interpretation
of the BMT inequality (or rather, a consequence thereof) in terms of the entropy of single-centered black holes. 
In \S\ref{sec_main}, we spell out the main theorem on which this work is based, giving an explicit relation between PT invariants and
D4-D2-D0 indices with at most two units of D4-brane charge.
In \S\ref{sec_X108}, we apply this tool to compute the first few terms
in the generating series of D4-D2-D0 indices with two units of 
D4-brane charge for two of the simplest models, namely $X_{10}$ and $X_8$, and identify 
in both cases a unique vector-valued mock modular form which
reproduces them. In \S\ref{sec_disc} we discuss our results and mention some open directions.
The mathematical proof of the main theorem and related results are presented in \S\ref{sec_proofs}.

\subsection*{Acknowledgements}
The authors are grateful to Nava Gaddam, Jan Manschot and Thorsten Schimannek for collaboration
on the previous works \cite{Alexandrov:2022pgd,Alexandrov:2023zjb}
which underlie the present one.
The research of BP is supported by Agence Nationale de la Recherche under contract number ANR-21-CE31-0021.
SF acknowledges the support of EPSRC postdoctoral fellowship EP/T018658/1.

\section{Brief review on BPS indices for Picard rank one CY threefolds}
\label{sec_review}

\subsection{Charge vectors and Donaldson-Thomas invariants}

In this work, we follow the same notations as in \cite{Alexandrov:2023zjb}. In particular,
$\CY$ is a smooth projective Calabi-Yau threefold with rank one Picard lattice $\Lambda=H^{1,1}(\CY)\simeq H^{2}(\CY,\IZ) = H \IZ$. We denote $\kappa=\int_{\CY} H^3$
the degree and $c_2=\int_{\CY} H.c_2(T\CY)$ the second Chern class of the tangent bundle.
The dual lattice $\Lambda^*=H^4(\CY,\IZ)$ is then generated by $H^2/\kappa$. 
Given a coherent sheaf $E$, we identity the Chern character $\ch(E)$ with the vector of rational numbers
\be
\label{defC0123}
[C_0,C_1,C_2,C_3](E) := \int_{\CY} [H^3 \ch_0(E)\,,\, H^2.\ch_1(E)\,,\, H.\ch_2(E)\,,\, \ch_3(E)] \in \IQ^4\,,
\ee
such that $\ch=(C_0+C_1 H + C_2 H^2+ C_3 H^3)/\kappa$.
Its components satisfy the following quantization conditions
\be
C_0\in \kappa \IZ,
\qquad
C_1\in  \kappa \IZ,
\qquad
C_2 \in \IZ+\frac{C_1^2}{2 \kappa},
\qquad
C_3 \in \IZ - \frac{c_2}{12\kappa}\, C_1.
\label{defCvec}
\ee
Similarly, we identity the electromagnetic charge (or Mukai vector) 
\be
\gamma(E):=\ch(E) \sqrt{\Td(T\CY)} = p^0 + p^1 H -  \frac{q_1}{\kappa}\, H^2
+ \frac{q_0}{\kappa}\,H^3
\ee
with the vector of rational numbers  $(p^0,p^1,q_1,q_0)$. 
The Chern and Mukai vectors are related by
\be
\label{Mukaibasis}
p^0 = \frac{C_0}{\kappa}\, ,
\qquad
p^1 = \frac{C_1}{\kappa}\, ,
\qquad
q_1
= - C_2 - \frac{c_2}{24\kappa}\, C_0,
\qquad
q_0= C_3+\frac{c_2}{24\kappa}\, C_1,
\ee
such that
\be
\label{quant}
p^0\in\IZ,
\qquad
p^1\in\IZ,
\qquad
q_1\in \IZ+\frac{\kappa}{2} \,(p^1)^2 - \frac{c_2}{24}\, p^0,
\qquad
q_0\in \IZ-\frac{c_2}{24}\, p^1.
\ee
In this basis, the Dirac pairing takes the Darboux form
\be
\langle \gamma, \gamma'\rangle 
:= \int_{\CY}
\gamma(E')^\vee \, \gamma(E) = q_0 p'^0 + q_1 p'^1 - q'_1 p^1 - q'_0 p^0\, .
\ee
Under tensoring with the line bundle  $\cO_\CY(k H)$, with Chern vector $[\kappa,\kappa k,\kappa k^2,\kappa k^3]$, 
the components of the Mukai vector transform as
\be
\label{specflowD6}
\begin{split}
p^0\mapsto p^0,
\qquad &
p^1\mapsto p^1+k p^0,
\qquad
q_1 \mapsto q_1 - \kappa k\, p^1  -\frac{\kappa k^2}{2}\, p^0\, ,
\\
&\,q_0 \mapsto q_0 -k q_1 +\frac{\kappa k^2}{2}\,  p^1 +
\frac{ \kappa k^3}{6}\, p^0 \, .
\end{split}
\ee
We  refer to this transformation as a spectral flow, and denote the transformed sheaf as $E(k)$. These definitions
and quantization conditions extend to objects $E=(\dots\rightarrow \cE^{-1}
\rightarrow \cE^{0} \rightarrow \cE^{1}\rightarrow \dots)$ 
in the derived category of coherent sheaves $\cC=D^b\Coh\CY$, with Chern
vector $\ch E =\sum_k (-1)^k \ch(\cE^k)$. As usual, we denote by $E[n]$
the shifted object $(\cE^{k-n})_{k\in\IZ}$, with Chern vector 
$\ch E[n]=(-1)^n \ch E$.

Given any charge vector $\gamma$ satisfying the quantization conditions \eqref{quant},  and any stability condition $\sigma=(Z,\cA)\in\Stab\cC$ (where $Z$ is a central charge function and $\cA$,
known as heart, is an Abelian subcategory of $\cC$ subject to certain conditions),
we denote by $\bOm_\sigma(\gamma)\in \IQ$ the (rational) DT invariant associated
to the moduli stack of semi-$\sigma$-stable objects of charge $\gamma$, and set 
\be
\label{defntilde}
\Omega_\sigma(\gamma) := \sum_{d|\gamma}  \frac{\mu(d)}{d^2}\,  \bOm_\sigma(\gamma/d),
\ee
where $\mu(d)$ is the Moebius function (equal to $+1$ ($-1$)  if $d$ is a square-free positive integer with an even (odd) number of prime factors, and zero otherwise). When the stability condition is generic with respect
to $\gamma$ (such that any  $\sigma$-semi-stable object is $\sigma$-stable),  
$\Omega_\sigma(\gamma)$ is conjecturally integer. If it furthermore lies 
along the slice of $\Pi$-stability condition, then $\Omega_\sigma(\gamma)$
is expected to coincide with the physical BPS index. The variation of $\bOm_\sigma(\gamma)$ with respect to $\sigma$ is governed by the standard wall-crossing formula \cite{Joyce:2008pc,ks}.

\subsection{Weak stability conditions and BMT inequality}
\label{sec_weak}

Rather than working with physical stability conditions (whose mathematical existence remains conjectural in general), it turns out to be convenient to use a family of weak stability conditions
with central charge function\footnote{The parameter $a$ is often
traded for $w=\frac12(a^2+b^2)$. then the weak stability condition is denoted by $\nu_{b,w}$. \label{fooabw}}
\be
Z_{b,a}(\gamma) = 
-a \(C_2-b C_1+\frac{b^2}{2} \,C_0\) + \frac{a^3}{2}\, C_0 + \I\, a^2 (C_1-b C_0) 
\ee
parametrized by $b+\I a\in \IH$, and heart $\cA_{b}$ obtained by tilting
the Abelian category $\Coh(\CY)$ with respect to the slope function $\mu_b=\frac{C_1}{C_0}-b$
(see \cite[\S 2]{Alexandrov:2023zjb} for more details). 
The advantage is that the chamber structure along this family is much simpler, with walls being nested semi-circles in the upper-half plane. 

Moreover, the conjectural BMT inequality (which lies at the basis of the construction of Bridgeland stability conditions) implies that the index $\bOm_{b,a}(\gamma)$ vanishes unless
\be
\label{BMTineq0}
L_{b,a}(\gamma) := \frac12 (C_1^2-2C_0 C_2) (a^2+b^2) + (3 C_0 C_3-C_1 C_2) b+(2C_2^2-3 C_1 C_3)\geq 0\, .
\ee
The condition $C_1^2-2C_0 C_2\geq 0$ is furthermore required by the classical Bogomolov-Gieseker inequality. In particular, the condition \eqref{BMTineq0}
rules out an open region above the real axis $a=0$ provided the discriminant 
of $L_{b,0}(\gamma)$ with respect to $b$ is strictly positive. This discriminant
is given by the quartic polynomial
\be
\label{defP4}
\begin{split}
P_4(\gamma) := &\,  (3 C_0 C_3-C_1 C_2) ^2 - 2  (C_1^2 -2 C_0 C_2) ( 2 C_2^2-3 C_1 C_3)
 \\
 =&\,
8 C_0 C_2^3  + 6 C_1^3 C_3 + 9 C_0^2 C_3^2  -3 C_1^2 C_2^2 - 18 C_0 C_1 C_2 C_3  \ .
\end{split}
\ee
For `small' charges, defined by the conditions 
\be
\label{condsmall}
C_1^2-2C_0 C_2> 0\, ,
\qquad 
P_4(\gamma)>0\, ,
\ee
the inequality \eqref{BMTineq0} therefore provides an empty chamber
in the vicinity of the real axis $a=0$ where the index $\bOm_{b,a}(\gamma)$ vanishes (assuming that the BMT inequality holds true). The DT invariants then can be obtained in other regions by applying the wall-crossing formulae (which also hold for weak stability conditions).

In order to interpret the condition $P_4(\gamma)>0$ physically, it is useful to adjust
the normalization of $P_4$ and rewrite it in terms of the electromagnetic charges:
\be
\label{P4I4}
-\frac{P_4(\gamma)}{9\kappa^2} 
=I_4(\gamma)
+\frac{c_2}{36\kappa} \(\kappa (p^1)^2 + 2 p^0 q_1 \)^2 
+ \frac{(c_2 p^0)^2}{432\kappa}  \(\kappa (p^1)^2 + 2 p^0 q_1 \)
+ \frac{c_2^3 (p^0)^4}{15552\kappa}\, ,
\ee
where 
\be
\label{defI4}
I_4(\gamma):=\frac{8 }{9 \kappa }\, p_0 q_1^3 -\frac{2\kappa}{3} \,  q_0 (p^1)^3  - ( p^0 q_0)^2
 +\frac{1}{3} \,(p^1 q_1)^2  -2 p^0 p^1  q_0 q_1
\ee
is the quartic polynomial controlling the existence 
of single-centered black holes in the large volume limit\footnote{By large volume limit, we mean the limit in which the charges $(p^0,p^1,q_1,q_0)$ scale as $(1,\lambda,\lambda^2,\lambda^3)$ with $\lambda\to\infty$, such that the attractor point lies at large volume.
Under this scaling, $I_4(\gamma)$ is homogeneous of degree 6, while the remaining 
terms in \eqref{P4I4} are of subleading degree 4, 2 and 0.} \cite{Shmakova:1996nz} (see for example \cite[(3.6)]{Manschot:2011xc}). 
Thus, so long as the subleading $c_2$-dependent terms in \eqref{P4I4} are smaller in absolute value than 
$|I_4(\gamma)|$, the condition $P_4(\gamma)>0$, which ensures the existence of an empty chamber in the family of weak stability conditions, is complementary to the condition $I_4(\gamma)>0$ which ensures
the existence of single-centered black holes in the large volume limit.\footnote{A similar interpretation in terms of the extremality bound for 5D black holes was proposed recently in \cite{Halder:2023kza}. It is equivalent to ours by virtue of the
4D/5D correspondence \cite{Gaiotto:2005gf}.} 
We leave it as
an open problem to understand physically the origin of the $c_2$-dependent terms in \eqref{P4I4}, or find an improved version of the BMT inequality where these contributions would be absent.\footnote{String theory does predict $c_2$-dependent corrections to the black hole entropy, but those should produce quadratic terms in $I_4(\gamma)$, rather than the quartic terms in \eqref{P4I4}.}

\subsection{DT, PT, GV invariants}

By exploiting wall-crossing along the family of weak stability conditions, it turns out that one can relate the rank 0 DT invariants counting D4-D2-D0 BPS states to rank $\pm 1$
invariants counting bound states of a single (anti) D6-brane with arbitrary D2-D0 charge.
Specifically, we define the standard DT invariant  \cite{MR1818182} counting ideal sheaves $E$ with $\ch(E)=1-\beta -m H^3/\kappa$ as 
\be
\label{noteDT}
{\mathrm I}_{m,\beta}=\DT(\beta. H,m)
=\Omega_{b<0,a\to+\infty}[\kappa,0,-\beta.H,-m]
\ee 
and the Pandharipande-Thomas (PT) invariant counting 
stable pairs \cite[\S 3]{Toda:2011aa}
$E=(\cO_\CY \stackrel{s}{\rightarrow} F)^\vee$ with $\ch(E)=-1+\beta-m H^3/\kappa$ as
\be
\label{notePT}
{\mathrm P}_{m,\beta}=\PT(\beta. H,m)
=\Omega_{b>0,a\to+\infty}[-\kappa,0,\beta.H,-m].
\ee
As indicated in \eqref{noteDT} and \eqref{notePT}, these two invariants coincide with 
the DT invariant $\Omega_{b,a}(\gamma)$ in the limit $a\to +\infty$, for suitable choice of the sign of $b$,
while square brackets indicate that the argument is expressed in terms of 
the Chern vector \eqref{defC0123}. Whereas the notations 
${\mathrm I}_{m,\beta}$, ${\mathrm P}_{m,\beta}$ have become standard in the mathematics literature, we prefer to use $\DT(Q,m)$, $\PT(Q,m)$ with scalar argument $Q=\beta.H\geq 0$.
In \cite{Alexandrov:2023zjb} (see also \cite{Liu:2022agh} for related results for the quintic $X_5$), 
we showed that these invariants vanish
unless $m$ satisfies the Castelnuovo-type bound
\be
\label{CastPT}
m\geq - \left\lfloor \frac{Q^2}{2\kappa} + \frac{Q}{2}\right\rfloor.
\ee

Given these lower bounds on $Q$ and $m$,
we can define the generating series
\be
\label{defZDT}
\begin{split}
Z_{DT} (y,\q):=&\, \sum_{Q,m} \DT(Q,m)\,
y^Q\,  \q^{m},
\\
Z_{PT} (y,\q):=&\, \sum_{Q,m} \PT(Q,m)\,
y^Q\,  \q^{m}.
\end{split}
\ee
In terms of these formal series,  the DT/PT relation conjectured
in \cite{pandharipande2009curve} and proven in \cite{toda2010curve,bridgeland2011hall} takes the simple form
\be
Z_{DT} (y,\q)= M(-\q)^{\chi_{\scriptstyle\CY}} \, Z_{PT} (y,\q),
\label{eqn:DTPTrelation}
\ee
where $M(\q)=\prod_{k>0}(1-\q^k)^{-k}$ is the Mac-Mahon function.
The MNOP relation~\cite{gw-dt,gw-dt2,Pandharipande:2011jz}  
then allows to compute $\DT(Q,m)$ and $\PT(Q,m)$
in terms of the GV invariants $\GV^{(g)}(Q')$ with $Q'\leq Q$ and $g\leq 1-m$. For our purposes, it is more convenient to use the plethystic form of this relation,
\be
\label{PTGVpleth}
Z_{PT}(y,\q)= \PE\[
\sum_{Q>0} \sum_{g=0}^{\gmax(Q)} (-1)^{g+1} \GVg{g}
 \(1-x\)^{2g-2} x^{(1-g)} y^Q\](-\q,y),
\ee
where $\PE$ denotes the plethystic exponential 
\be
\PE[f](x,y)=\exp\(\sum_{k=1}^\infty\frac{1}{k}\, f(x^k,y^k)\).
\ee
The GV invariants $\GVg{g}$ count embedded curves $C$ of genus $g$ and class $[C]=Q H^2/\kappa$ \cite{Gopakumar:1998ii,Gopakumar:1998jq}. As a consequence of the PT/GV relation \eqref{PTGVpleth}, they vanish unless
\be
\label{defgmax}
g\leq g_{\rm max}(Q) :=  \left\lfloor \frac{Q^2}{2\kappa} + \frac{Q}{2}\right\rfloor + 1,
\ee
or equivalently
\be
\label{defqmin}
Q \geq Q_{\rm min}(g) :=  \left\lceil -\frac{\kappa}{2} + \frac12 \sqrt{\kappa(8(g-1)+\kappa)}\right\rceil .
\ee
For genus $g=0$, the invariants $\GVg{0}$ are computable by mirror symmetry techniques and are equal to the DT invariants $\Omega_\infty(0,0,Q, m)$ counting 
stable coherent sheaves supported on $C$ with arbitrary Euler class $m\in \IZ$ such that 
$(Q, m)$ are coprime. For higher genus $g\geq 1$, they are computable by 
integrating the holomorphic anomaly equations for the topological string partition function
(a method also known as `direct integration') \cite{Huang:2006hq}.

\subsection{D4-D2-D0 indices}
\label{sec_modconj}

Finally, we come to the main invariants of interest, namely the rank 0 DT invariants 
counting D4-D2-D0 bound states.
For a fixed charge $\gamma$ with $p^0=0$,  the index $\bOm_{b,a}(\gamma)$ 
reaches a finite value  as $a\to +\infty$,
denoted by $\bOm_{\infty}(\gamma)$, which  
coincides with the index
$\bOmH(\gamma)$ counting Gieseker-semi-stable sheaves.
For CY threefolds with Picard rank one, this index also agrees with the `large volume attractor index' 
(also called MSW index in 
\cite{Alexandrov:2012au,Alexandrov:2016tnf,Alexandrov:2017qhn,Alexandrov:2018lgp})
\be
\label{lvolatt}
\bOm_{\infty}(0,r,q_1,q_0) = \lim_{\lambda\to +\infty}
\bOmPi_{-\frac{q_1}{\kappa r}+\I \lambda  r}(0,r,q_1,q_0)\,,
\ee
where $\bOmPi_{\rmz=b+\I t}(\gamma)$ denotes the DT invariant along the $\Pi$-stability slice
(see \cite[\S 2]{Alexandrov:2023zjb} for precise definitions of these objects). 
The index $\bOm_{\infty}(0,r,q_1,q_0)$ is preserved
under spectral flow \eqref{specflowD6} with $k\in\IZ$, which leaves
the D4-brane charge $r$ invariant, as well
as the reduced D0-brane charge
\be
\label{defqhat}
\hq_0 :=
q_0 -\frac{q_1^2}{2\kappa r}
\ee
and the class of $\mu:=q_1 - \frac12\, \kappa r^2$ in $\Lambda^*/\Lambda=\IZ/(\kappa r \IZ)$. Accordingly, we denote
\be
\label{defOmrmu}
\bOm_{r,\mu}(\hq_0)=\bOm_\infty(\gamma)=\bOmH(\gamma)\, ,
\ee
and refer to these invariants as charge $r$ D4-D2-D0 indices (or Abelian D4-D2-D0 indices, when $r=1$).

The index \eqref{defOmrmu} is invariant under
$\mu\mapsto -\mu$ and periodic under $\mu\mapsto\mu+\kappa r$. 
Furthermore, 
$\Omega_{r,\mu}(\hq_0)$ vanishes unless  the reduced charge $\hq_0$
is bounded from above by \cite[Lemma B.3]{Feyzbakhsh:2021rcv}
\be
\hq_0 \leq \hq_0^{\rm max}:=\frac{1}{24}\, \chi(\cD_r),
\label{qmax}
\ee
where $\chi(\cD_r)$ is the topological Euler characteristic of the divisor $\cD_r$ Poincar\'e dual to $rH$ \cite[(3.8)]{Maldacena:1997de}
\be
\chi(\cD_r) = \kappa r^3+c_{2}r.
\label{defchiD}
\ee
Note that  the reduced charge $\hq_0$ is such that the combination
\be
\label{nfromhq0}
n = \frac{\chi(\cD_r)}{24} -\frac{\mu^2}{2\kappa r} - \frac{r \mu}{2} - \hq_0
\ee
is an integer. Due to this, it will be convenient to introduce yet another notation for D4-D2-D0
indices
\be
\tOm_{r,\mu}(n) := \bOm_{r,\mu}\left( \frac{\chi(\cD_r)}{24} -\frac{\mu^2}{2\kappa r} - \frac{r \mu}{2} - n \right),
\label{def-tOm}
\ee
which now satisfies 
\be
\tOm_{r,\mu}(n) = \tOm_{r,-\mu}(n+r\mu)
=\tOm_{r,\mu+k\kappa r}
\left(n-k\mu-\frac{1}{2}\, k\kappa r(r+k) \right),
\qquad
k\in\IZ.
\label{prop-tOm}
\ee
In particular, these relations imply
\be
\tOm_{r,\mu}(n) =\tOm_{r,-\mu-\kappa r^2}(n).
\label{prop-tOm2}
\ee

Since the reduced D0-brane charge is bounded from above for fixed D4-brane charge $r>0$ and D2-brane charge 
$q_1=\mu+\frac12\, \kappa r^2$, one can define the generating series of rational invariants
\be
h_{r,\mu}(\tau) :=\sum_{\hq_0 \leq \hq_0^{\rm max}}
\bOm_{r,\mu}(\hq_0)\,\q^{-\hq_0 }
= \q^{\frac{\mu^2}{2\kappa r} + \frac{r \mu}{2} -\frac{\chi(\cD_r)}{24}}
 \sum_{n\geq - \frac{\mu^2}{2\kappa r} - \frac{r \mu}{2} }
\tOm_{r,\mu}(n)\,\q^{n},
\label{defhDT}
\ee
where $\q=e^{2\pi\I \tau}$. 
Since $\mu$ takes values in  $\IZ/(\kappa r \IZ)$, \eqref{defhDT}
defines a vector with $\kappa r$ entries (half of which being redundant due to the symmetry under $\mu\mapsto -\mu$). 
For $r=1$, the case studied in \cite{Alexandrov:2023zjb}, the charge vector was 
necessarily primitive and there was no distinction between 
the `rational' DT invariant $\bOm_{1,\mu}(\hq_0)$ and the `integer' DT invariant
$\Omega_{1,\mu}(\hq_0)$.
In the case $r=2$ of interest in this paper, it is crucial to use the rational DT invariants in \eqref{defhDT} in order
to get nice (mock) modular properties \cite{Manschot:2010xp}.

\subsection{Modular properties of generating functions}

By exploiting the constraints of S-duality in string theory, it was argued 
in \cite{Alexandrov:2012au,Alexandrov:2016tnf,Alexandrov:2018lgp,Alexandrov:2019rth} 
that the generating series $h_{r,\mu}$  must possess specific modular properties
under $SL(2,\IZ)$ transformations of the parameter $\tau$.
More precisely, $h_{r,\mu}$ should transform  as a weakly holomorphic vector valued mock modular form of depth $r-1$, with a specific modular anomaly. This means that there is
an explicit non-holomorphic completion $\widehat{h}_{r,\mu}$, determined in terms of the
generating series $h_{r',\mu'}$ with $0<r'<r$, such that $\widehat{h}_{r,\mu}$ transforms
as a standard vector valued modular form of weight $-3/2$ 
with the following matrices for $T:\tau\mapsto \tau+1$ and $S:\tau\mapsto -1/\tau$
\cite[Eq.(2.10)]{Alexandrov:2019rth} (see also \cite{Gaiotto:2006wm, deBoer:2006vg, Denef:2007vg, Manschot:2007ha})
\be
\begin{split}
M_{\mu \nu}(T)=&\,
e^{\frac{\pi\I}{\kappa r}(\mu+\frac{\kappa r^2}{2} )^2 +\frac{\pi\I r}{12}\, c_{2} }
\,\delta_{\mu\nu},
\\
M_{\mu\nu}(S)=&\,
\frac{(-1)^{\chi_{r\cD}}}{\sqrt{\I \kappa r}}\,
e^{-\frac{2\pi\I}{\kappa r}\,\mu \nu},
\end{split}
\label{multsys-h2}
\ee
where 
\be
\label{defL0}
\chi_{r\cD}:=\chi(\cO_{r\cD})=\frac{\kappa r^3}{6}+\frac{r c_2}{12}
\ee
and $\mu,\nu\in \IZ/(\kappa r \IZ)$. These properties characterize the Weil representation attached to the lattice
$\IZ[\kappa]$ with quadratic form
$m\mapsto \kappa r m^2$, up to an overall factor of the multiplier system of $\eta^{c_2 r}$
where $\eta$ is the Dedekind eta function. We denote by $\scM_r(\CY)$ the space of weakly holomorphic vector-valued modular forms with these transformation properties under $SL(2,\IZ)$. Its dimension was computed in
\cite[\S A]{Alexandrov:2022pgd}, and is equal to $n_r-C_r$, where $n_r$ is the number of polar terms (which grows like $r^4$) and
$C_r$ is the number of modular constraints\footnote{These constraints
come from holomorphic cusp forms in dual weight, see \cite{Bantay:2007zz,Manschot:2008zb} for details.} (which grows linearly in $r$). The numbers $n_r$ and $C_r$ for $r=1,2$
are tabulated in Table \ref{table1}.

For $r>1$, $h_{r,\mu}$ is {\it not} an element of $\scM_r(\CY)$ (due to the modular anomaly), but one could add to it  
an element $\hh_{r,\mu}\in \scM_r(\CY)$ without affecting its transformation properties. In other words, in order to determine $h_{r,\mu}$, it suffices to find
a holomorphic $q$-series $h^{({\rm an})}_{r,\mu}$ with the same anomalous modular transformation as $h_{r,\mu}$, and fix the modular ambiguity $h^{(0)}_{r,\mu}$ by computing at least $n_r-C_r$ coefficients in the $q$-expansion of the $h_{r,\mu}$'s. 

The first step in this procedure was carried out for $r=2$ in \cite[\S 5]{Alexandrov:2022pgd}. There it was shown that $h_{2,\mu}$ can be decomposed as 
\be
\begin{split}
h_{2,\mu}
=&\, \hh_{2,\mu}+(-1)^{\mu+\kappa} \sum_{\mu_1=0}^{\kappa-1}
\Gi{\kappa}_{\mu-2\mu_1+\kappa}\, h_{1,\mu_1}\,h_{1,\mu-\mu_1}\,,
\end{split}
\label{defnorm-G}
\ee
where $\hh_{2,\mu}\in \scM_2(\CY)$, $h_{1,\mu_1}$ is the generating series of Abelian D4-D2-D0 invariants, and $\Gi{\kappa}_\mu$ 
is a vector valued mock modular form with a specific shadow depending only on the degree $\kappa$. For $\kappa=1$, $\Gi{\kappa}_\mu$ can be chosen to coincide 
with the generating series $H_\mu$ of Hurwitz class 
numbers \cite{Zagier:1975}
\be
\Gi{1}_\mu=H_\mu,
\label{solG1}
\ee
with the following $q$-expansion 
\be
\label{H01}
\begin{split}
H_0(\tau) =&\, -\frac{1}{12} + \frac{\q}{2}+\q^2+\frac{4 \q^3}{3}+\frac{3 \q^4}{2}+2 \q^5+2 \q^6+2 \q^7+3 \q^8+\frac{5
\q^9}{2}+2 \q^{10}
+\dots\, ,
\\
H_1(\tau) =&\, \q^{\frac34} \left( \frac{1}{3}+\q+\q^2+2
\q^3+\q^4+3 \q^5+\frac{4 \q^6}{3}+3 \q^7+2 \q^8+4 \q^9+\q^{10}
 + \dots \right).
\end{split}
\ee
For $\kappa>1$ and prime we take instead
\be
\Gi{\kappa}_\mu=\hf\, (V_\kappa[H])_\mu,
\label{GkapH}
\ee
where $V_\kappa$ is a generalized Hecke-like operator introduced in \cite{Bouchard:2016lfg,Bouchard:2018pem}
\be
(V_\kappa[\phi])_\mu(\tau)=\kappa\sum_{a,d>0 \atop ad=\kappa}d^{-2}\sum_{b=0}^{d-1}
\delta^{(1)}_{\mu/a} \, e^{\frac{\pi\I b}{2 a\kappa}\,\mu^2}
\phi_{\mu/a} \(\frac{a\tau+b}{d}\),
\label{defHecke-exp}
\ee
where $\delta_x^{(n)}=1$ if $x=0\mod n$ and 0 otherwise.
Thus, we only need to fix the holomorphic modular ambiguity $\hh_{2,\mu}\in \scM_2(\CY)$. 

An overcomplete basis for $\scM_2(\CY)$ can be constructed as in \cite{Alexandrov:2022pgd}, by introducing the
vector-valued, weight $1/2$ theta series
\be
\label{theta2}
\tvths{\kappa}_{\mu}(\tau):=
 \sum_{{k}\in \IZ+\frac{\mu}{2\kappa}}
\q^{\kappa\,k^2} 
\ee
and by multiplying by sufficient powers of the Dedekind function $\eta$, Eisenstein series $E_4$ and $E_6$, and the Serre derivative $D$ (acting as $q\partial_q-\frac{w}{12}E_2$ on modular forms of weight $w$)
so as to obtain the desired modular weight and leading power of $\q$:
\be
\hh_{2,\mu}(\tau) =
\sum_{\ell=0}^{\mm}
\sum_{k=0}^{k_\ell} a_{\ell,k}\,
E_4^{\lfloor w_\ell/4 \rfloor-\eps_\ell-3k}(\tau)\, E_6^{2k+\eps_\ell}(\tau)\,
\frac{D^\ell \tvths{\kappa}_{\mu}(\tau)}{\eta^{32\kappa + 2c_2}(\tau)}\, ,
\label{decomp-modform2}
\ee
Here the integers $w_\ell, k_\ell, \eps_\ell$  are given by
\be
w_\ell=16\kappa+c_2-2-2\ell,
\qquad 
k_\ell=\lfloor w_\ell/12 \rfloor-\delta^{(12)}_{w_\ell-2},
\qquad
\eps_\ell=\delta^{(2)}_{w_\ell/2-1},
\ee
while $\mm$ should be chosen sufficiently large so that
$\sum_{\ell=0}^{\mm}(k_\ell+1)$ is not smaller than the
dimension of the space $\scM_2(\CY)$. The coefficients $a_{\ell,k}$
can be fixed by computing the first few terms in the $\q$-expansion
of $\hh_{2,\mu}$.

\section{$r=2$ D4-D2-D0 invariants from wall-crossing}
\label{sec_main}

In this section, we spell out a Theorem,  extending \cite[Thm 1]{Alexandrov:2023zjb}, which expresses stable pair invariants $\PT(Q,m)$ in terms
of stable pair invariants $\PT(Q',m')$ of lower degree $Q'<Q$ {\it and} D4-D2-D0 invariants $\tOm_{r,\mu}(n)$ with $r\leq 2$. This result is proven in Appendix \ref{sec_proofs} by studying wall-crossing for a rank $-1$ class $\v$ in the space of weak stability conditions introduced in \S\ref{sec_weak}. 
In \S\ref{subsec-mainth} we rewrite it using notations \eqref{notePT} and 
\eqref{def-tOm} which make it easier in applications.
In \S\ref{subsec-exprD4} we reformulate the resulting formula
so that it expresses D4-D2-D0 invariants $\tOm_{2,\mu}(n)$ with $r=2$ units of D4-brane charge 
in terms of PT invariants. Once the generating series of $\tOm_{2,\mu}(n)$ is found, 
it also provides new boundary conditions for the direct integration method,
as we explain in \S\ref{sec_bcmock}.

\subsection{Main theorem}
\label{subsec-mainth}

Let us define the function $f_2:\IR^+\to \IR$ by\footnote{Note that in the range $x<4$, $f_2(x)$ coincides with the function $f(x)$ defined in \cite[(4.10)]{Alexandrov:2023zjb}, which we denote by $f_1(x)$ in this work.}
\be
\label{deffx}
f_2(x)\ \coloneqq \ \left\{\!\!\!\begin{array}{cc} x+ \frac{1}{2} & \text{if $0 <x < 1$,}
\\
\vspace{.1 cm}
\sqrt{2x+\frac{1}{4}} & \text{if $1 \leq x < \frac{15}{8}$}\, ,
\\
\vspace{.1 cm}
\frac{2}{3}\,x+ \frac{3}{4} & \text{if $\frac{15}{8} \leq x< \frac{9}{4}$}\, ,
\\
\vspace{.1 cm}
\frac{1}{3}\,x+ \frac{3}{2} & \text{if $\frac{9}{4} \leq x< 3$}\, ,
\\
\vspace{.1 cm}
\frac{1}{2}\,x+ 1 & \text{if $3 \leq x <4$}\, ,
\\
\vspace{.1 cm}
\frac{3}{2}\,\sqrt{x} & \text{if $4 \leq x <9$}\, ,
\\
\vspace{.1 cm}
\frac{1}{3}\, x+ \frac{3}{2}  & \text{if $x \geq 9$}\, .
\end{array}\right.
\end{equation}
Further, for any $(Q,m)\in \IZ_+\times \IZ$, let 
\be
\label{defxa}
x:=\frac{Q}{\kappa}\,,
\qquad
\alpha:=-\frac{3m}{2Q}\,.
\ee
Theorem \ref{thm-rank2} in Appendix \ref{sec_proofs} states that whenever $x>0$ and $f_2(x)<\alpha$,
the stable pair invariant $\PT(Q,m)$ is expressed in terms of invariants
$\PT(Q',m')$ with $Q'<Q$ and D4-D2-D0 invariants $\tOm_{r,\mu}(n)$ with $r\leq 2$.
Physically, this result follows by studying wall-crossing for 
the charge vector $\gamma=\(-1,0,-Q-\frac{c_2}{24},-m\)$ corresponding to an anti-D6-brane. 
The precise formula reads 
\be
\PT(Q,m) = P_1(Q,m) + P_2(Q,m) + P_3(Q,m),
\ee
where the three contributions come from the anti-D6-brane emitting a single Abelian D4-brane, 
a pair of Abelian D4-branes, or a single charge 2 D4-brane, respectively,
and are given as follows:\footnote{Translating equations in \S\ref{subsec-thm} 
into the form written below, we took into account that 
$\dt(0,rH,\beta,m)=\tOm_{r,\mu}(n)$ where $\mu=-\beta.H-\hf\kappa r^2$
and $n=\frac16 \kappa r^3-m$. Furthermore, we repeatedly used the property \eqref{prop-tOm2}.}

\subsubsection{Contribution from a single Abelian D4-brane:}

\be
\label{thmS11}
P_1 =   \sum_{(Q', m')} (-1)^{\chi_1(Q', m')}
\chi_1(Q', m') \, \PT_{w(Q')}(Q',m') \, \tOm_{1,Q-Q'}(m-m'+Q')\,,
\ee 	
where $\chi_1(Q', m')=m -m'+ Q+Q'-\chiOD$ with $\chiOD$ defined in \eqref{defL0}. 
The sum runs over pairs of integers $(Q', m')$ such that
\be
\begin{split}
\label{Qpbound}
0 \leq &\; Q'  \leq Q + \kappa \left(\frac12-\alpha\right),
\\
-\frac{Q'^2}{3\kappa} - \frac{Q'}{3} w(Q')
\leq&\; m' \leq  m+\frac{1 }{2\kappa}\, (Q-Q')^2
+ \frac{1}{2}\,(Q+Q') \, ,
\end{split}
\ee
where $w(Q')=\frac{1}{\kappa}\, (2Q-3Q')-1$. This contribution is similar to \cite[Thm 1]{Alexandrov:2023zjb}, except that $\PT(Q',m')$ 
on the right-hand side
is replaced by the invariant $\PT_{w(Q')}(Q',m')$ counting tilt-semistable objects at 
$(b,w)=(2,w(Q'))$ in the space of weak stability conditions (see footnote \ref{fooabw}). 
The latter is in turn expressed in terms of the standard PT and Abelian D4-D2-D0 invariants as follows (Proposition \ref{lem.rank -1} in the Appendix):
\be
\begin{split}
\PT_w(Q,m)
=&\,\PT(Q,m)-
\sum_{Q',m'} (-1)^{\chi_1(Q', m')} \chi_1(Q', m') \PT(Q',m')\, \tOm_{1,Q-Q'}(m-m'+Q'),
\end{split}
\ee
where the sum runs over pairs of integers $(Q',m')$ such that 
\be
\begin{split}
0 \leq &\;  Q'  <\frac{\kappa}{2}\, (1-w) + \frac{Q}{2}\, ,
\\
-\frac{Q'^2}{2\kappa} - \frac{Q'}{2} 
\leq &\; m'  \leq
m+\frac{1}{2\kappa} \,(Q-Q')^2 + \frac{1}{2}\,(Q+Q') .
\end{split}
\ee

\subsubsection{Contribution from a pair of Abelian D4-branes:}
\be
\begin{split}
P_2 = &\,  \sum_{Q',n,n'}
(-1)^{\chi_2(Q',n)+\chi_2(Q',n')+1}\chi_2(Q',n)\,\chi_2(Q',n')\, 
\tOm_{1,Q'}(n) \, \tOm_{1,Q'}(n')\\
&\, \qquad
\times
\PT(Q-2Q'+3\kappa,m-n-n'+2Q-6Q'+5\kappa),
\end{split}
\ee
where $\chi_2(Q',n)=n+Q+Q'-\chiOD$. 
The sum runs over triplets of integers $(Q',n,n')$ such that
\be
\begin{split}
\kappa\(\alpha+\hf\)
& \le  Q' \le \hf\, (Q+3\kappa),
\\
n,n' & \ge -\frac{Q'^2}{2\kappa}-\frac{Q'}{2}\, ,
\\
n +n' &\leq  m+\frac{(Q-2Q')^2}{2\kappa}+\frac{11}{2}\, Q- 13Q'+11\kappa .
\end{split}
\ee

\subsubsection{Contribution from a single charge 2 D4-brane:}
\be
\begin{split}
P_3 
=&\,  \sum_{Q',m'}  (-1)^{\chi_3(Q',m')}\chi_3(Q',m')\,
\PT(Q',m')\, \tOm_{2,Q-Q'}(m-m'+2Q'),
\end{split}
\ee
where $\chi_3(Q',m') = m-m'+2(Q+Q') -\chi_{2\cD}$. The sum
runs over pairs of integers $(Q',m')$ such that 
\be
\begin{split}
0 \leq &\;  Q'  \leq 
Q+2\kappa(1-\alpha),
\\
-\frac{Q'^2}{2\kappa}  - \frac{Q'}{2}
\leq &\; m' \leq
m+\frac{1}{4\kappa}\,(Q-Q')^2 +Q+Q'.
\label{P3-cond}
\end{split}
\ee

\subsection{Computing charge 2 D4-D2-D0 indices from PT invariants}
\label{subsec-exprD4}

We first remark that for $x<4$, or for $x\geq 4$ but $\alpha>\frac12 x+1$, the result above  reduces to Theorem 1 in \cite{Alexandrov:2023zjb}. In particular, one may check that $P_2=P_3=0$ and $\PT_{w(Q')}(Q,m)=\PT(Q,m)$ in the formula for $P_1$. As explained in loc. cit., one can then isolate the contribution from $(Q',m')=(0,0)$
in $P_1$ to express the Abelian D4-D2-D0 index $\tOm_{1,Q}(m)$ in terms of PT invariants and Abelian D4-D2-D0 indices of lower degree.

Similarly here, we observe that when $\alpha\leq \frac12 x+1$, the range of $Q'$ in \eqref{P3-cond} allows for the contribution of $Q'=0$, in which case the second condition in \eqref{P3-cond} reduces to $0\leq m'\leq m+\frac{Q^2}{4\kappa}+Q=(-\frac23\alpha+\frac{x}{4}+1)Q$. Since the only non-vanishing invariant
$\PT(Q',m')$ with $Q'=0$ is $\PT(0,0)=1$, whenever 
$x>4$ and $\alpha\leq \frac38 x+\frac32$, the term $P_3$ includes a contribution proportional to the charge 2 D4-D2-D0 invariant,
namely
\be
P_3 = (-1)^{\chi_3(0,0)} \chi_3(0,0)\, \tOm_{2,Q}(m) + P_3',
\ee
where 
\be
P'_3 =  \sum_{Q'>0,m'}  (-1)^{\chi_3(Q',m')}\chi_3(Q',m')\,
\PT(Q',m')\, \tOm_{2,Q-Q'}(m-m'+2Q').
\ee
Assuming that $\chi_3(0,0)=m+2Q -\chi_{2\cD}$ is non-zero, we can thus compute the charge 2 D4-D2-D0 invariant via
\be
\label{WCF-rank2-main}
\tOm_{2,Q}(m)=\frac{(-1)^{m-\chi_{2\cD}}}{m+2Q -\chi_{2\cD}}\Bigl[\PT(Q,m) - P_1 -P_2 -P'_3\Bigr],
\ee
where the right-hand side involves PT invariants and D4-D2-D0 indices of lower degree. 

Moreover, as explained in \cite{Alexandrov:2023zjb} and elaborated in Remark \ref{rkvk} in the Appendix, we can start from the DT invariant $\tOm_{2,Q}(n)$ with arbitrary charges
and apply the spectral flow \eqref{specflowD6} with $k$ sufficiently large\footnote{Since the PT invariants $\PT(Q,m)$
are only known up to some limited degree $Q_{\rm max}$ (determined by the maximal genus $g_{\rm max}$ for which GV invariants are known
via $Q_{\rm max}=Q_{\rm min}(g_{\rm max})$), it is advisable
to use the smallest value $k_0$ such that the 
conditions \eqref{condk} are satisfied for $k\geq k_0$; occasionally, we observe 
(and prove in a few cases) 
that the formula in the theorem continues to hold for one unit less
of spectral flow; sometimes we also get non-trivial checks by applying the theorem with one extra unit of spectral flow, see e.g. \eqref{X8test}. \label{fook0}
} 
such that the resulting charges 
\be
Q_k=\mu+2\kappa k,
\qquad
m_k=n -k\mu -\kappa k(k+2)=\frac{\chi(2\cD)}{24} -\frac{Q_k^2}{4\kappa } - Q_k - \hq_0
\ee
satisfy the conditions 
\be
f_2(x_k)<\alpha_k\leq \frac38 x_k+\frac32\, ,
\label{condk}
\ee
so that $\tOm_{2,Q}(n)$ can be computed by
applying \eqref{WCF-rank2-main} for the shifted charges $(Q_k,m_k)$. This provides an efficient way of computing $\tOm_{2,Q}(n)$ (or equivalently $\bOm_{2,Q}(\hat q_0)$), provided the PT invariants
(or equivalently the GV invariants) are known to sufficient degree and Euler class (or genus). 

\subsection{Boundary conditions from mock modularity}
\label{sec_bcmock}

As explained in \cite[\S 3.2]{Alexandrov:2023zjb}, the topological
string partition function (hence the GV and PT invariants) can be computed
inductively by solving the holomorphic anomaly equations, provided sufficiently many conditions are known to determine the holomorphic
ambiguity at each genus. Currently, the known conditions include the usual conifold gap constraints, the Castelnuovo bound \eqref{defgmax}, 
and the value of GV invariants saturating the bound when $Q=0\mod \kappa$, 
\be
\GV_{\kappa d}^{(1+\frac12 \kappa d(d+1))}=(-1)^{\frac12 \kappa d(d-1)}
\chi_{\cD} \left( \chi_{\cD} + \frac12  \kappa d(d-1) \right)
\label{valGVkappa}
\ee
for any $d\geq 2$ \cite[(3.31)]{Alexandrov:2023zjb},
where $\chi_{\cD}$ was defined in \eqref{defL0}. These constraints 
are sufficient to fix the holomorphic ambiguity so long as
\be
\label{Deltag}
\Delta(g):= \left\lfloor \frac{2(g-1)}{\rho} \right\rfloor 
- \left\lfloor -\frac{\kappa}{2} + \frac12 \sqrt{\kappa(8(g-1)+\kappa)}\right\rfloor \leq 0,
\ee
where $\rho$ is the effective regulator ($\rho=d$ for the
hypersurfaces $X_{d}$ of degree $d=5,6,7,8$). In \eqref{Deltag}, the first term corresponds to the number of coefficients in the holomorphic ambiguity at genus $g$ taking into account the gap conditions and constraints \eqref{valGVkappa}, while the second term is the number of constraints from the Castelnuovo bound. Without
further conditions, the direct integration method can thus be applied
up to genus $g_{\rm integ}$, the highest value of $g$ such that $\Delta(g)\leq 0$. Ignoring the floor
function, we have $g_{\rm integ}\simeq \frac12\kappa \rho(\rho-1)+1$.

To push the direct integration method to higher genus, we need to provide $\Delta(g)$ additional
constraints at each genus $g$, for example the invariants $\GV^{(g)}_{Q_{\rm min}(g)+k}$
for $k=1,\dots,\Delta(g)$. The largest degree occurs for $k=\Delta(g)$ leading to 
\be
Q_\Delta(g)= Q_{\rm min}(g)+\Delta(g)=\left\lfloor \frac{2(g-1)}{\rho}\right\rfloor.
\label{relQg}
\ee
Let us now introduce a distance away from the Castelnuovo bound (normalized by $\kappa$)
\be
\delta(Q,g):=\frac{1}{\kappa}\,(g_{\rm max}(Q)-g),
\label{defdist}
\ee
where $g_{\rm max}$ is given in \eqref{defgmax}.
The distance corresponding to \eqref{relQg} and expressed in terms of the variable 
$x$ defined in \eqref{defxa} is found to be
\be
\label{deltag}
\delta_\Delta(x)\simeq \delta(\kappa x,\kappa \rho x/2+1)
\simeq 
\hf\, x(x-\rho+1).
\ee
On the other hand,
assuming that the generating series of Abelian D4-D2-D0 invariants are known, the PT invariants $\PT(Q,m)$, and hence the GV invariants $\GV^{(g\leq 1-m)}_{Q'\leq Q}$, 
can be computed using \cite[Thm 1]{Alexandrov:2023zjb} so long as 
$\alpha>f_1(x)$ with $\alpha$ as in \eqref{defxa}.
Let us rewrite this condition in terms of $\delta$.
To this end, we note that for fixed $Q$ and $m$, the distance \eqref{defdist} corresponding to the maximal reachable genus is
\be
\delta(\kappa x,1-m)\simeq \frac{m}{\kappa}+\frac{Q^2}{2\kappa^2}+\frac{Q}{2\kappa}
=-\frac23\, x \alpha + \frac12\, x(x+1).
\ee
Therefore, one arrives at the condition that 
\be
\label{deltax1}
\delta < \delta_1(x) := -\frac23\, x f_1(x)+\frac12\, x(x+1).
\ee
In the range $x>3$, one has $f_1(x)=\frac12 x+1$ so that $\delta_1(x)=\frac16 x(x-1)$.
Then it is easy to see that the curves $\delta_\Delta(x)$ 
and $\delta_1(x)$ intersect at $x=\frac32\rho-2$. Using \eqref{relQg}, this gives
\be
\label{gmod}
g_{\rm mod}^{(1)}\simeq\frac34\, \kappa \rho^2 - \kappa \rho+1, 
\ee
which is the maximal genus attainable with the results of our
previous work \cite{Alexandrov:2023zjb}. 

Similarly, if one assumes that the charge 2 D4-D2-D0 invariants are also known, then we
can compute the PT invariants $\PT(Q,m)$ using Theorem \ref{thm-rank2} so long as 
\be
\label{deltax2}
\delta < \delta_2(x) := -\frac23\, x f_2(x)+\frac12\, x(x+1).
\ee
In the range $x>9$, one has $f_2(x)=\frac13 x+\frac32$ so that $\delta_2(x)=\frac{1}{18}x(5x-9)$. The curves $\delta_\Delta(x)$ 
and $\delta_2(x)$ now intersect at $x=\frac94(\rho-2)$ which gives the higher value for the maximal attainable genus
\be
\label{gmmod}
g_{\rm mod}^{(2)}\simeq \frac98\, \kappa \rho^2 - \frac94\, \kappa \rho+1 .
\ee
This analysis ignores boundary effects due to the floor function 
appearing in $x=\frac{1}{\kappa}\left\lfloor \frac{2(g-1)}{\rho}\right\rfloor$
(hence the $\simeq$ sign in \eqref{gmod} and \eqref{gmmod}),
but nonetheless predicts the correct maximal genus attainable 
with the knowledge of the generating series of charge 2 D4-D2-D0 
invariants, except for two odd points at genus 70 and 77 for $X_6$, see Figure \ref{figXGV}.

It is interesting to ask whether computing D4-D2-D0 invariants at higher 
D4-brane charge might provide sufficiently many new boundary conditions to extend the direct integration method to arbitrary genus. Of course, if one knew all higher
rank D4-D2-D0 invariants, one could use \cite[Thm 2]{Feyzbakhsh:2020wvm,Feyzbakhsh:2021rcv} to compute any PT invariant, and hence any GV invariant. However, in practice one can only hope to determine
higher charge D4-D2-D0 invariants inductively, by using an analogue of Theorem \ref{thm-rank2} for computing the polar terms, and then exploiting mock modularity constraints to obtain the full generating series at charge $r$. 
As explained in Remark
\ref{rem_higherrk}, a formula expressing PT invariants in terms of 
D4-D2-D0 invariants up to charge $r$  is expected to hold if the values $b_1,b_2$, at which the BMT wall \eqref{BMTineq0} intersects the horizontal axis $a=0$, satisfy the inequalities $b_2-b_1>b_1$ and $r\leq b_1<r+1$ (see Remark
\ref{rem_higherrk} 
in the Appendix). These inequalities can be recast as 
\be
\begin{cases} 
\frac32\,\sqrt{x}  < \alpha \leq \frac{x}{r}+\frac{r}{2}  & \mbox{if} \ r^2< x < (r+1)^2, \\
\frac{x}{r+1} + \frac{r+1}{2} < \alpha \leq \frac{x}{r}+\frac{r}{2}  & \mbox{if}  \ x>(r+1)^2.
\end{cases}
\ee
Hence the formula would hold whenever $f_r(x)<\alpha$ where $f_r(x)=\frac{x}{r+1} + \frac{r+1}{2}$ in the range $x>(r+1)^2$. Running the same arguments as before, we find that the GV invariants would be determined up to 
\be
\label{gmmodr}
g_{\rm mod}^{(r)}\simeq \frac38\, \kappa \rho^2 (r+1) - \frac14\, \kappa \rho (r+1)^2+1 ,
\ee
reproducing the previous conditions for $r=1$ and $r=2$, respectively.
This formula peaks at a finite value $r+1=\frac34\rho$ leading to 
$g_{\rm mod}^{(r)}=1+\frac{9}{64}\,\kappa\rho^3$, which appears to be the maximal genus attainable by this method. For $X_{5,6,8,10}$, it evaluates to 
$89,92,145,141$, respectively.

\begin{figure}[h]
\begin{center}
\includegraphics[height=5cm]{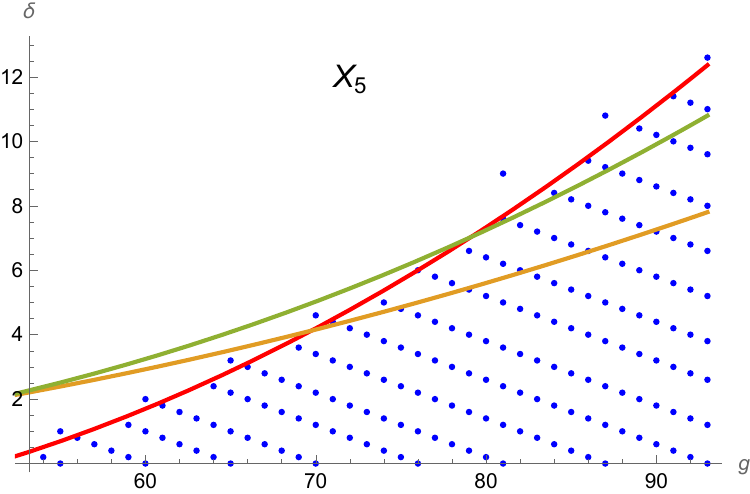}
\hspace*{5mm}
\includegraphics[height=5cm]{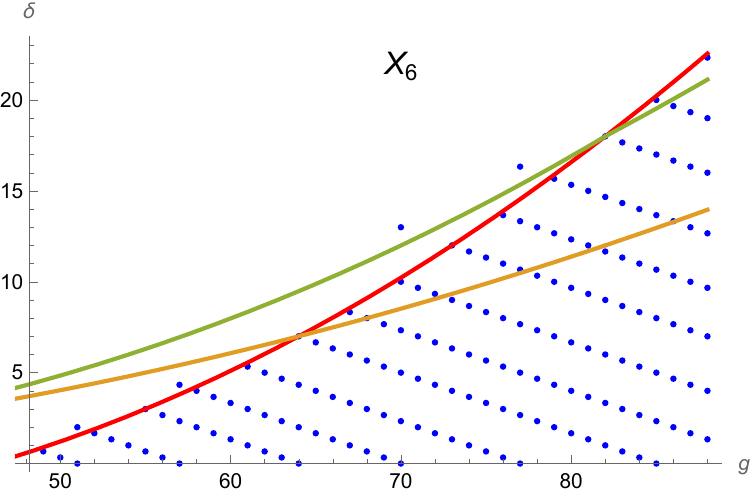}
\\
\includegraphics[height=5cm]{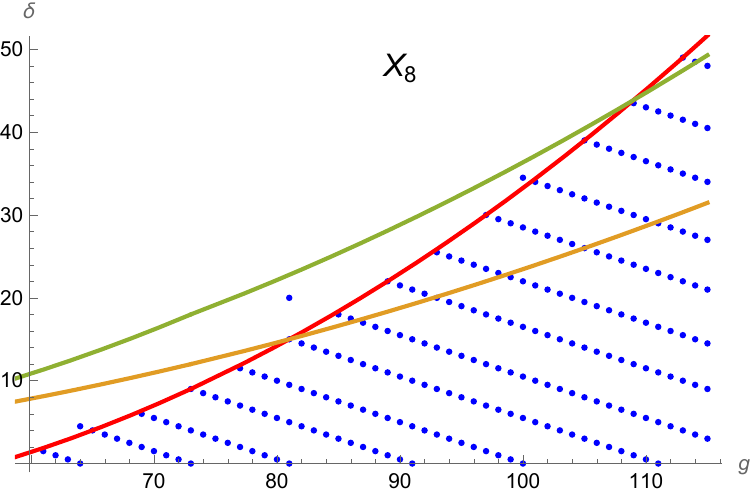}
\hspace*{5mm}
\includegraphics[height=5cm]{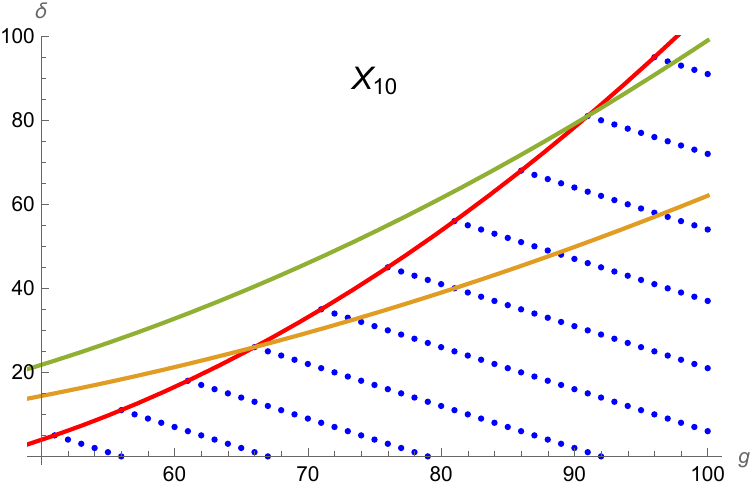}
\end{center}
\caption{Additional GV invariants needed to fix the holomorphic ambiguity for $X_5$, $X_6$, $X_8$ and $X_{10}$. The vertical axis corresponds to 
the distance away from the Castelnuovo bound normalized as in \eqref{defdist}.
The blue dots are pairs $(g,\delta(Q,g))$ 
where $Q=Q_{\rm min}(g)+k$ with $1\leq k\leq \Delta(g)$.
The red, orange and green lines correspond to 
$\delta_\Delta(x)$, 
$\delta_1(x)$ and $\delta_2(x)$ with $x=2(g-1)/(\kappa\rho)$.}
\label{figXGV}
\end{figure}

\section{Mock modular forms for $X_{10}$ and $X_8$}
\label{sec_X108}
In this section, we determine the first few terms in the generating series of D4-D2-D0 indices with two units of D4-brane charge for the simplest hypergeometric models, namely the decantic $X_{10}$ and the octic $X_8$, 
and for each of them find a unique mock modular form of depth one which
matches these coefficients.

\subsection{$X_{10}$}
\label{sec_X10}

We first consider the degree 10 hypersurface in weighted projective space 
$\IP^{5,2,1,1,1}$. In this case, the basic divisor $H$ has self-intersection $\kappa=H^3=1$. In \cite[\S 5.3]{Alexandrov:2023zjb}, using knowledge of the GV invariants up to genus 47 and assuming the validity of the BMT conjecture (which has not yet been established mathematically for this case), we computed the first twelve coefficients in the generating series of Abelian D4-D2-D0 indices, and found a perfect match with the 
weak modular form conjectured in \cite{VanHerck:2009ww} (correcting the earlier guess in \cite{Gaiotto:2007cd})
\be
\begin{split}
h_1=&\, \frac{203 E_4^4+445 E_4 E_6^2}{216\, \eta^{35}}
\\
=&\, \q^{-\frac{35}{24}}\,\Bigl(
\underline{3-575 \q}+271955 \q^2+206406410 \q^3+21593817025 \q^4
+\cdots\Bigr).
\end{split}
\label{resfunX10}
\ee
Here and below, we underline the polar coefficients. Using this result, 
we obtained additional boundary conditions for the direct integration method, which allow us to compute GV invariants in principle up to genus 70. In practice, due to limited computational power, we have currently been able to compute GV invariants up to genus 67. 

For two units of D4-brane charge, the generating series $h_{2,\mu}$ is now a two-dimensional vector with 7 polar terms. Applying the formula \eqref{WCF-rank2-main} with the minimal value of spectral flow satisfying the conditions \eqref{condk} (see footnote \ref{fook0}), we find that the leading terms are given by\footnote{In order to obtain the $\cO(\q^4)$ coefficient in $h_{2,1}$, we used \cite[Prop. 1]{Alexandrov:2023zjb} to compute the PT invariant 
$\PT(13,-65)=599652917084283118466718009161486$ which enters this coefficient.}
\be
\begin{split}
h_{2,0} =&\, \q^{-\frac{19}{6}} \left( \underline{7 -1728 \q +203778 \q^2 -13717632 \q^3}
+\dots \right),
\\
h_{2,1} =&\, \q^{-\frac{35}{12}}  \left( \underline{ -\tfrac{21}{4}+1430\q-\tfrac{4344943}{4}  \q^2}+
 208065204 \q^3
  -\tfrac{199146131237}{4}  \q^4
+ \dots\right),
\end{split}
\label{predictX10}
\ee
Here, the first four coefficients in $h_{2,0}$ arise by summing up 6, 16, 34 and 62 walls corresponding to 5+0+1, 15+0+1, 33+0+1 and 61+0+1 contributions in $P_1+P_2+P_3$, respectively. Similarly, the  first five coefficients in $h_{2,1}$ arise by summing up 9+1+1, 23+2+1, 46+3+1, 51+4+2 and 86+5+2 walls.\footnote{A Mathematica notebook performing these computations is available on the webpage \cite{CYdata}.} The chamber structure for the rank $-1$ charges used to extract the first terms 
in $h_{2,0}$ and $h_{2,1}$ is depicted in Figure \ref{figChamberX10}.

\begin{figure}[h]
\begin{center}
\includegraphics[height=6cm]{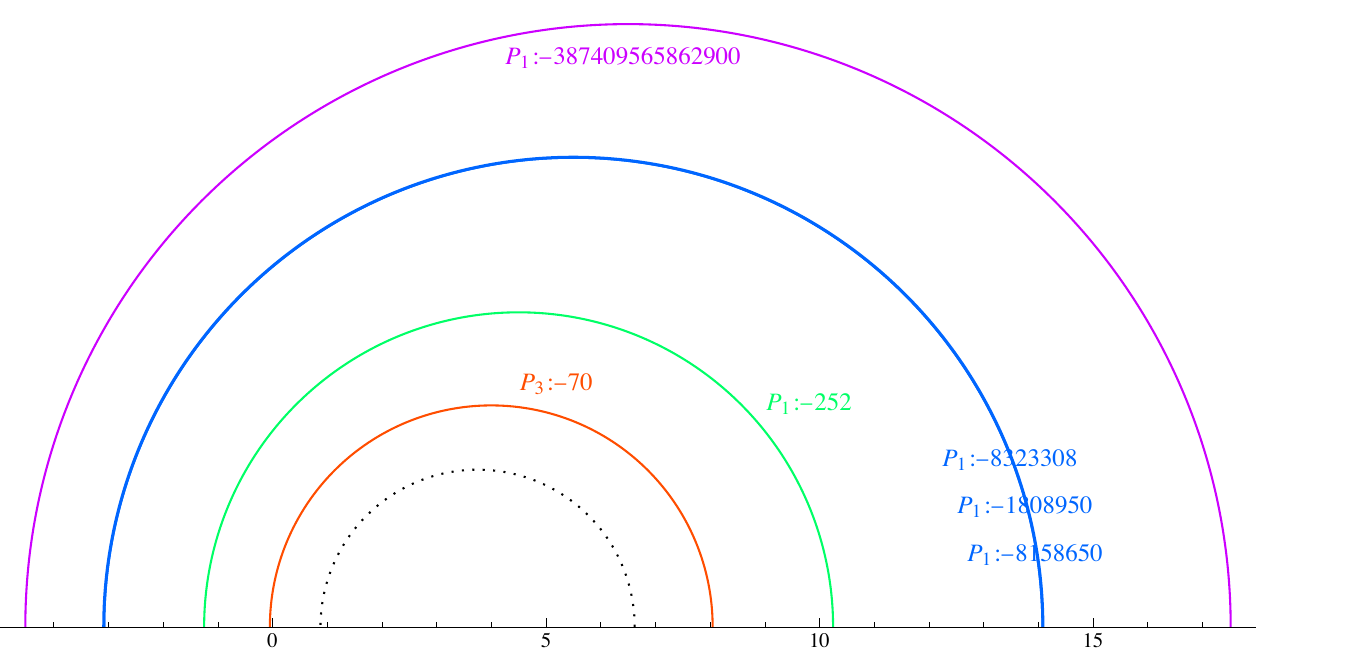}\\[5mm]
\includegraphics[height=6cm]{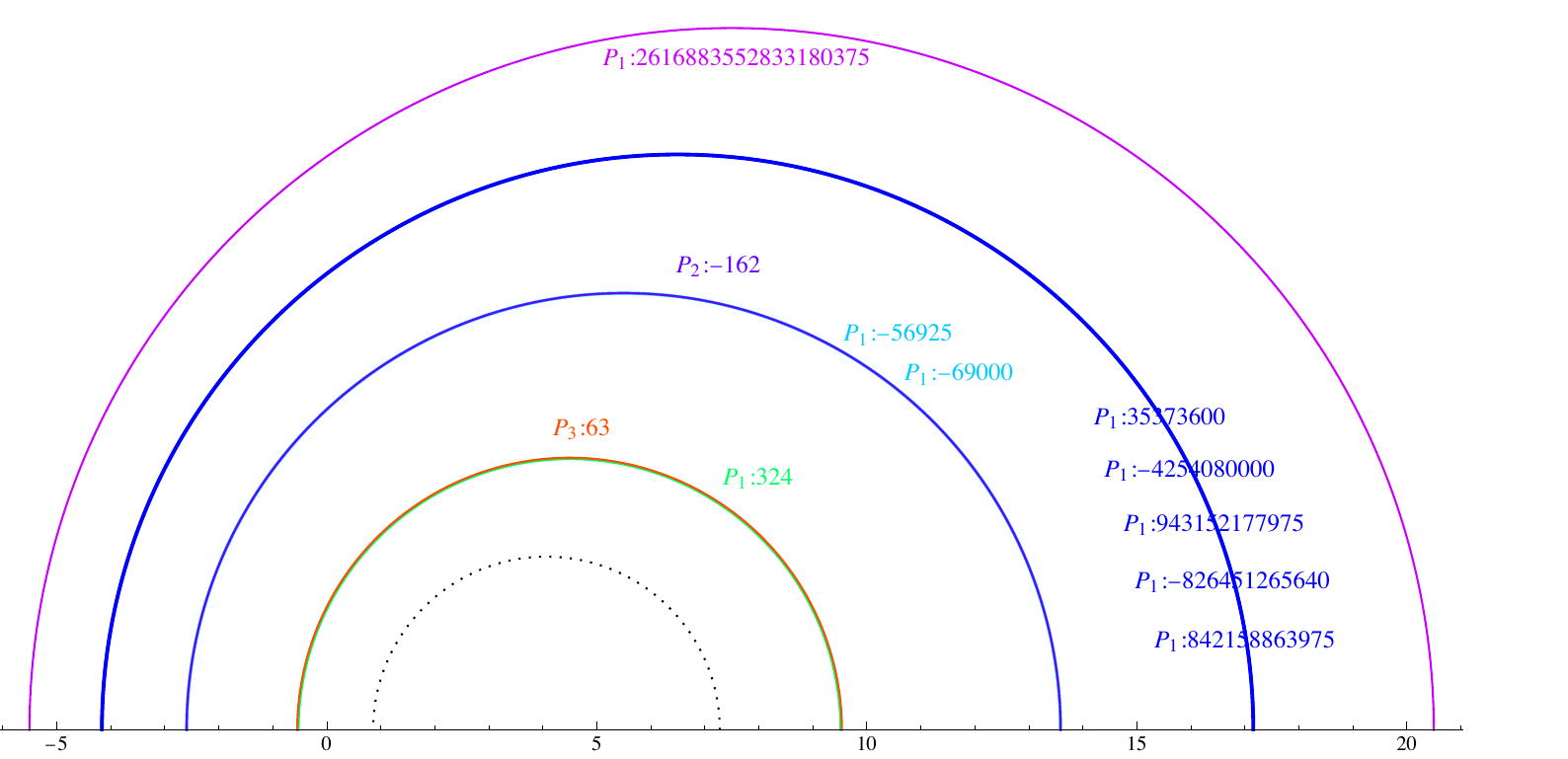}
\end{center}
\caption{Chamber structure for the charge vectors $\gamma=[-1,6,0,15]$ (top) and
$\gamma=[-1,7,0,19]$ (bottom) used to determine the first coefficients in $h_{2,0}$ and $h_{2,1}$ for $X_{10}$. 
The contributions of the walls add up to 
$\PT(6,-15)=-387409584154130$ and $\PT(7,-19)=2616884507474124585$, respectively. The dotted line delimits the empty region predicted by the BMT inequality. The orange wall (which overlaps with a $P_1$-type wall on the bottom graph) corresponds to the emission of a charge 2 D4-brane.}
\label{figChamberX10}
\end{figure}

While we used the chamber structure (in the space of weak stability conditions) of an anti-D6-brane charge with suitable D2 and D0 charges
to obtain  the coefficients \eqref{predictX10}, it is worth observing 
that the leading term in $h_{2,0}$ can be interpreted more directly as the contribution of a \DDb bound state $\gamma=\gamma_1+\gamma_2$ with $\gamma_1=\cO_{\CY}(-2H)[1]$ and $\gamma_2=\cO_{\CY}$, such that $\langle\gamma_1,\gamma_2\rangle=7$. The second coefficient $-1728$ 
also agrees with the naive ansatz $-6 \DT(0,1)$ from \cite{Alexandrov:2022pgd}, while the third coefficient differs from 
the naive ansatz $5 \DT(0,2)=203760$ by only 18 units. 
Similarly, the leading term
in $h_{2,1}$ can be interpreted as the contribution of a bound state $\gamma=2\gamma_1+2\gamma_2$, with  $\gamma_1=\cO_{\CY}(-2H)[1]$ and $\gamma_2=\cO_{\CY}(-H)$, such that $\langle\gamma_1,\gamma_2\rangle=3$; indeed, the rational index
 of a rank $(2,2)$ Kronecker quiver with 3 arrows\footnote{
 More generally, the rational index of a
rank $(2,2)$ Kronecker quiver with $m$ arrows is given by 
$\bOm_{K_m}(2,2)=-\frac12 m^3+m^2-\frac14 m$
\cite[(5.38)]{Beaujard:2019pkn}. \label{fookro}
} is $\bOm_{K_3}(2,2)=-\frac{21}{4}$. 
It would be very interesting to determine directly 
the chamber structure associated to the polar terms in \eqref{predictX10}.

Remarkably, there exists a unique mock modular form which matches the
expansion \eqref{predictX10}: 
\bea
h_{2,\mu}& =&
\frac{5397523 E_4^{12}+70149738 E_4^9 E_6^2-12112656 E_4^6 E_6^4
-61127530 E_4^3 E_6^6-2307075 E_6^8} {46438023168 \eta^{100}}
\tvths{1}_\mu
\nn\\
&&  +
\frac{-10826123 E_4^{10} E_6-14574207 E_4^7 E_6^3+20196255 E_4^4 E_6^5
+5204075 E_4 E_6^7}{1934917632 \eta^{100}}
D \tvths{1}_\mu 
\label{genX10mod}\\ 
&& + (-1)^{\mu+1} H_{\mu+1}(\tau) \, h_{1}(\tau)^2,
\nn
\eea
where $H_\mu$ is the generating series of Hurwitz class numbers \eqref{H01}.
This predicts the following  rational DT invariants:
\be
\begin{split}
h_{2,0} =&\, \q^{-\frac{19}{6}} \left( \underline{7 -1728 \q +203778 \q^2 -13717632 \q^3}
-23922034036 \q^4 
+ 5622937903758 \q^5\right. \\ & \left. 
+30869696301784635 \q^6 + 11647656420016054694 \q^7 + 
 1947303902022148176051 \q^8+ \dots \right),
\\
h_{2,1} =&\, \q^{-\frac{35}{12}}  \left( \underline{ -\tfrac{21}{4}+1430\q-\tfrac{4344943}{4}  \q^2}+
208065204 \q^3 -\tfrac{199146131237}{4}  \q^4 +168558189742220 \q^5+\right.\\
&\left. +\tfrac{385592585541997181}{2} \q^6 + 53445664437544786684 \q^7 
+ \tfrac{29798457431526946287313}{4} \q^8+  \dots \right)
\end{split}
\label{predictX10mod}
\ee
and the following integer DT invariants, which indeed appear to be integers at high order in $\q$,
\be
\begin{split}
\hint_{2,1} =&\, h_{2,1}(\tau) -\frac14\, h_1(2\tau) \\
=& \, \q^{-\frac{35}{12}}  \left( \underline{ -6+1430\q-1086092 \q^2}+ 208065204 \q^3
-49786600798 \q^4 +168558189742220  \q^5 \right. \\
& \left. +192796292719396988 \q^6 + 53445664437544786684 \q^7 + 
 7449614357876338117572 \q^8
+ \dots\right),
\end{split}
\ee
with $\hint_{2,0} = h_{2,0}(\tau)$. 

Given that the relevant space of vector-valued modular forms 
has dimension 7, 
the fact that we can reproduce two extra non-polar coefficients and 
that integer DT invariants are indeed integer provides very strong support 
for the generating function \eqref{genX10mod}. Additional checks come from the fact that the $\cO(\q^4)$, $\cO(\q^5)$ terms in $h_{2,0}$ and the $\cO(\q^4)$ term in $h_{2,1}$ can be obtained by applying \eqref{WCF-rank2-main}
with $k_0-1$ units of spectral flow, in the notation of footnote \ref{fook0} (this leads to the cases $(x,\alpha)=(12,\frac{11}{2})$, $(12,\frac{43}{8})$ and $(11,\frac{111}{22})$ covered by Lemma \ref{LemmaX6}).

\subsection{$X_8$}
\label{sec_X8}
We now consider the octic in $\IP^{4,1,1,1,1}$. In this case, the basic divisor $H$ has self-intersection $\kappa=H^3=2$. 
We note that 
a weaker version of the BMT conjecture was
established for this space in \cite{koseki2022stability}, which is sufficient for the proof of Theorem \ref{thm-rank2} to go through.
In \cite[\S B.2]{Alexandrov:2023zjb}, using knowledge of the GV invariants up to genus 48, we computed 18 coefficients in the generating series of Abelian D4-D2-D0 indices, and found a perfect match with the 
weak modular form conjectured in \cite{Gaiotto:2007cd}
\be
\label{hmuX8}
\begin{split}
h_{1,\mu}=&\, \frac{1}{\eta^{52}} \[\frac{103 E_4^6+1472 E_4^3 E_6^2+153 E_6^4}{5184}
+\frac{503 E_4^4 E_6+361 E_4 E_6^3}{108}\, D\]\vths{2}_{\mu}.
\end{split}
\ee
Using this result, 
we obtained additional boundary conditions for the direct integration method, which allow us to compute GV invariants in principle up to genus 84, in practice up to genus 64. 

For two units of D4-brane charge, the generating series $h_{2,\mu}$ is now 
a three-dimensional\footnote{More precisely, it is a four-dimensional vector, 
but one component is redundant since $h_{2,3}=h_{2,-1}=h_{2,1}$.} 
vector with 14 polar terms, with one modular constraint. 
Applying the formula \eqref{WCF-rank2-main}, we find that the leading terms are given by
\be
\label{X8h2predict}
\begin{split}
h_{2,0} =&\, \q^{-\frac{13}{3}} \left( \underline{-10 +2664 \q -344564 \q^2 + ? \q^3 +? \q^4}
+ \dots \right),
\\
h_{2,1} =&\, \q^{-\frac{101}{24}}  \left( \underline{0+0\q+206528\q^2 -51750016 \q^3 +?\q^4}
+ \dots\right),
\\
h_{2,2} =&\, \q^{-\frac{23}{6}}
\left(\underline{-17 +? \q  +? \q^2 + ? \q^3}
+ \dots\right).
\end{split}
\ee
Here, the first coefficients in $h_{2,0}$ arise by summing up 12+0+1, 45+0+1 and 51+0+1 walls, respectively. Similarly, the  first coefficients in $h_{2,1}$ arise by summing up 17+0+1, 57+0+1, 64+0+1 and 71+0+1 walls. The first coefficient in $h_{2,2}$
arises by summing up 26+1+1 walls (see Figure \ref{figChamberX8} for the chamber structure relevant for the first coefficients in $h_{2,0}$, $h_{2,1}$ and $h_{2,2}$). The question marks indicate coefficients which cannot be determined from our current knowledge of GV invariants 
upon using 
the minimal value $k_0$ of spectral flow satisfying the condition \eqref{condk}.

\begin{figure}[h]
\begin{center}
\includegraphics[height=6cm]{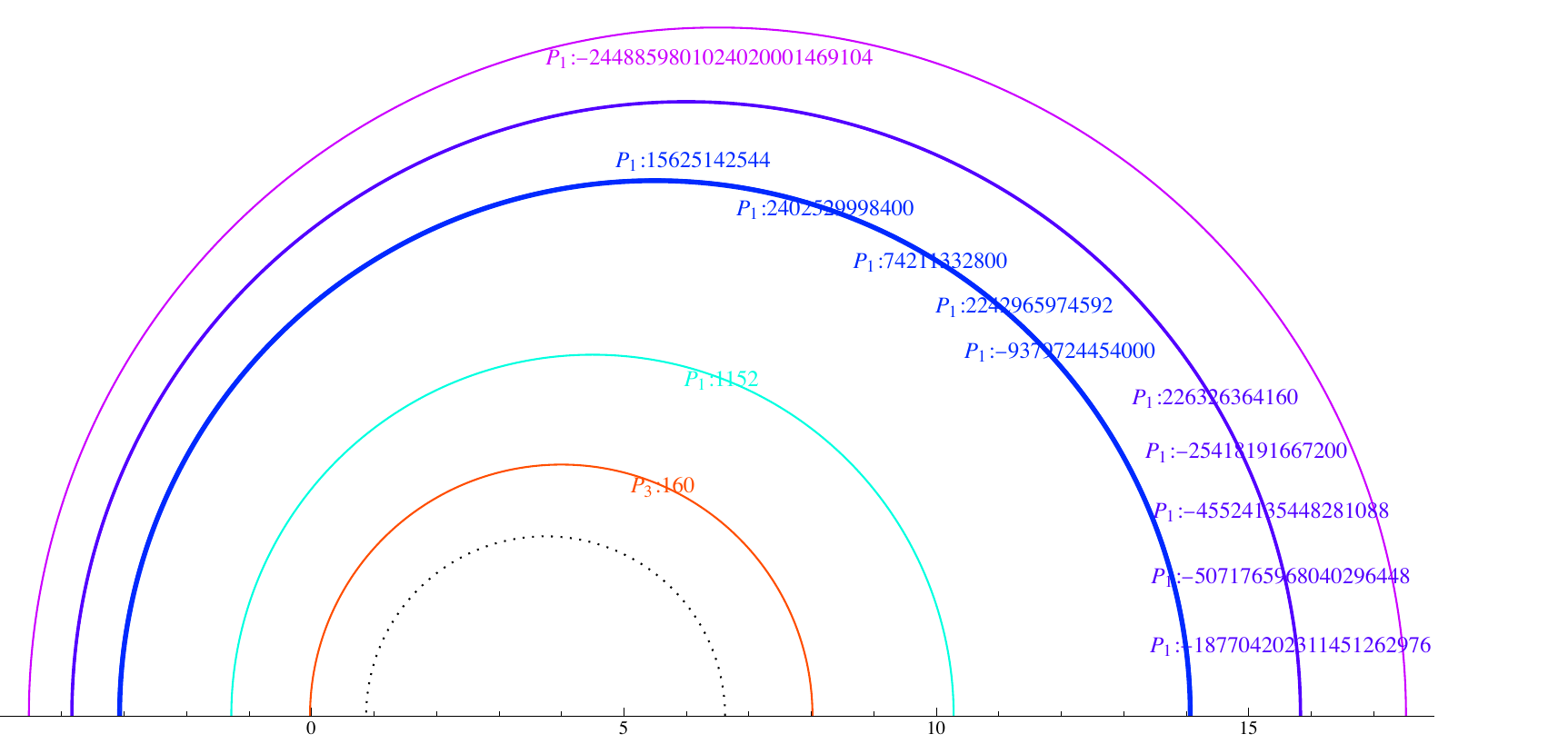}\\[3mm]
\includegraphics[height=6cm]{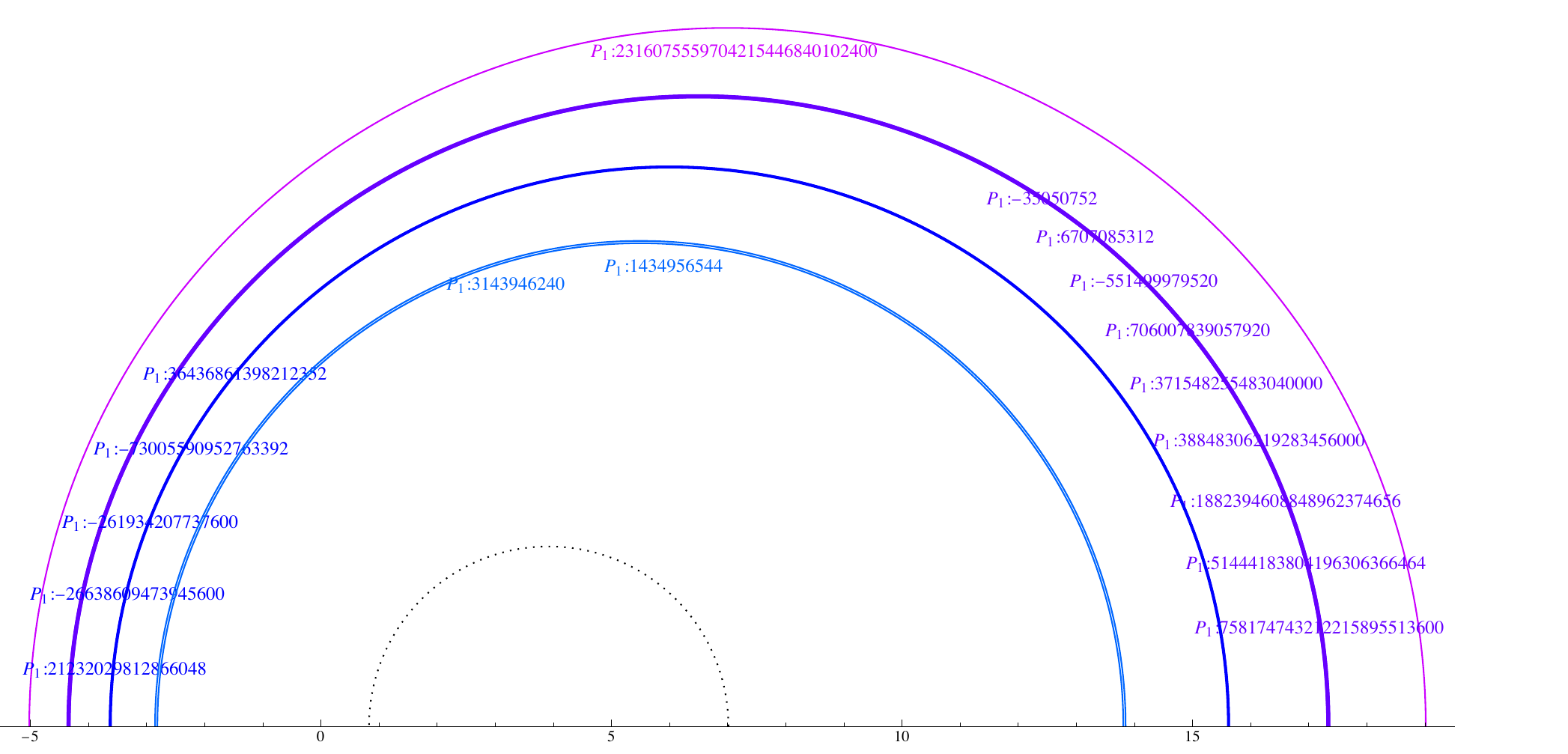}\\[3mm]
\includegraphics[height=6cm]{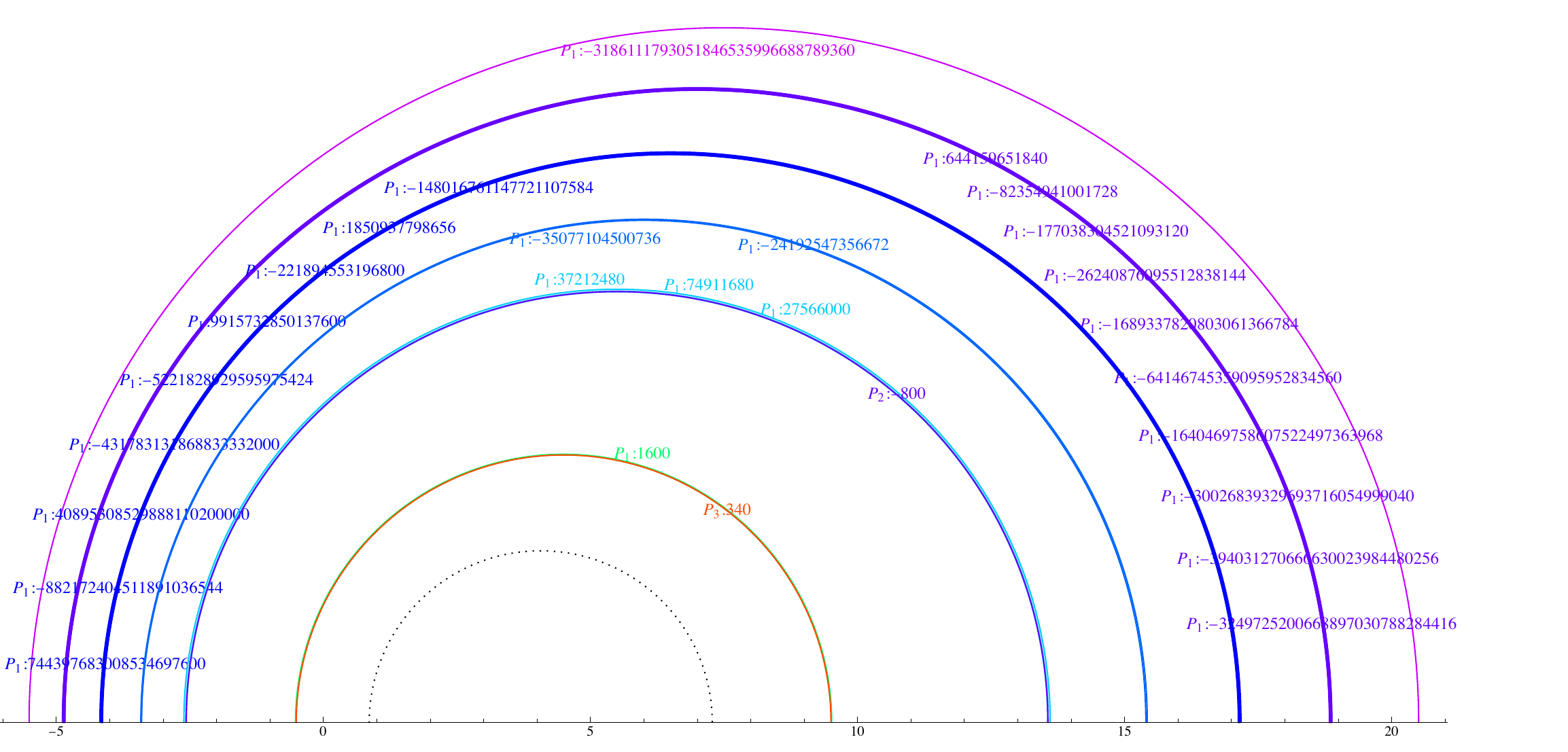}
\end{center}
\caption{Chamber structure for the charge vectors $\gamma=[-2,12,0,30]$,
$\gamma=[-2,13,0,34]$ (middle) and
$\gamma=[-2,14,0,38]$ (bottom) used to determine the first coefficients in $h_{2,0}$, $h_{2,1}$ and $h_{2,2}$ for $X_8$. The orange wall is absent on the middle graph, consistent with the vanishing of the first coefficient in $h_1$. 
The contributions of the walls add up to 
$\PT(12,-30)=-2449052622546271198617008$, $\PT(13,-34)=2316887100204163406937500672$ and $\PT(14,-38)=-3189787242522249817441517380140$, respectively.}
\label{figChamberX8}
\end{figure}

As for $X_{10}$, we observe that the leading term in $h_{2,0}$ can be interpreted as the contribution of a \DDb bound state $\gamma=\gamma_1+\gamma_2$ with $\gamma_1=\cO_{\CY}(-2H)[1]$ and $\gamma_2=\cO_{\CY}$, such that $\langle\gamma_1,\gamma_2\rangle=10$.
The second coefficient $2664$ 
also agrees with the naive ansatz $9 \DT(0,1)$ from \cite{Alexandrov:2022pgd}, while the third coefficient differs from 
the naive ansatz $8 \DT(0,2)=-344544$ by only 20 units. 
Similarly, the first non-zero coefficient in $h_{2,1}$
agrees with the naive ansatz $7 \DT(1,1)$, and the leading term
in $h_{2,2}$ can be interpreted as the contribution of bound state $\gamma=2\gamma_1+2\gamma_2$, with  $\gamma_1=\cO_{\CY}(-2H)[1]$ and $\gamma_2=\cO_{\CY}(-H)$, such that $\langle\gamma_1,\gamma_2\rangle=4$ (see footnote \ref{fookro}). It would be very interesting to have a similar interpretation of the other terms in \eqref{X8h2predict}, or even better in \eqref{exph2} below. 

While \eqref{WCF-rank2-main} based on Theorem \ref{thm-rank2} does not allow directly to compute the coefficients marked with $?$, 
Lemmas \ref{LemmaX81} and \ref{LemmaX82} show that for the $\cO(\q^3)$ coefficient in $h_{2,0}$,
the $\cO(\q^4)$ and $\cO(\q^5)$ coefficients in $h_{2,1}$, and
the $\cO(\q)$ coefficient in $h_{2,2}$, the formula \eqref{WCF-rank2-main} in fact continues to hold with $k_0-1$ units of spectral flow. 
Moreover, Lemma \ref{LemmaX83} states that the $\cO(\q^4)$
coefficient in $h_{2,0}$  can also be computed using $k_0-1$ units of spectral flow, provided one also includes contributions from bound states of the form
$\gamma=2\gamma_1+\gamma_2$ where $\gamma_1=\cO_{\CY}(3H)[1]$ and $\gamma_2=(1,0,-H^2,H^3)\otimes \cO_{\CY}(6H)$, with $\bOm(\gamma_1)=1$ and 
$\bOm(\gamma_2)=\DT(2,-2)=6$. Since $\langle\gamma_1,\gamma_2\rangle=16$ and
$\bOm_{K_{16}}(2,1)=120$, we find
\be
\tOm_{2,0}(4)=-\frac{29672563324}{11} +\frac{1}{22} \times 120 \times 6 =
-2697505724,
\ee
where the first term corresponds to the contribution given by \eqref{WCF-rank2-main}. 
Similarly, Lemma \ref{LemmaX84} asserts that the $\cO(\q^2)$
coefficient in $h_{2,2}$ can also be computed using $k_0-1$ units of spectral flow, provided one also includes contributions from bound states of the form $\gamma_0+\gamma_4$,
where $\gamma_4=\cO_{\CY}(5H)$ and
$\gamma_0=3\gamma_1+\gamma_2$, while $\gamma_1=\cO_{\CY}(3H)[1]$ and $\gamma_2=\cO_{\CY}(4H)$.
Since $\langle\gamma_1,\gamma_2\rangle=4$, $\langle \gamma_0,\gamma_4\rangle=26$ and
$\bOm_{K_{4}}(3,1)=-4$, we get
\be
\tOm_{2,2}(0)
=-\frac{9373102}{9} -\frac{1}{18} \times (-4) \times 26 =-1041450,
\ee
where the first term again corresponds to the contribution from \eqref{WCF-rank2-main}.
Using these results, we arrive at
\be
\begin{split}
h_{2,0} =&\, \q^{-\frac{13}{3}} \left( \underline{-10 +2664 \q  -344564\,\q^2 + 28739232 \q^3-2697505724\, \q^4}+\dots \right),
\\
h_{2,1} =&\, \q^{-\frac{101}{24}}  \left( \underline{0+0\q+206528\q^2 -51750016 \q^3 +6162707008 \q^4}
+ 3811912230976 \q^5 + \dots\right),
\\
h_{2,2} =&\, \q^{-\frac{23}{6}}
\left(\underline{-17+6216 \q -1041450\,\q^2 + ?\, \q^3}
+ \dots\right).
\end{split}
\ee
Remarkably, there exists a unique vector-valued mock modular solution which matches the
14 coefficients which we have computed:
\be
\begin{split}
h_{2,\mu} =&\, - \frac{E_4^2 E_6}{69657034752 \eta^{104}}\,
   \Bigl(27489243085 E_4^9+189114116937 E_4^6 E_6^2 
\\
&\, +177632766471 E_4^3
   E_6^4+21771164387 E_6^6\Bigr) \, \tvths{2}_\mu   
\\
&\,+\frac{1}
   {5804752896\eta^{104}}\,
   \Bigl(4132634222 E_4^{12}+90332924791 E_4^9 E_6^2+173271553815
   E_4^6 E_6^4
\\
&\, +50628642229 E_4^3 E_6^6+895654223
   E_6^8\Bigr) \, D  \tvths{2}_\mu
\\
&\, -\frac{ E_4 E_6}{241864704  \eta^{104}}\,
   \Bigl(4768083767 E_4^9+25127775147
   E_4^6 E_6^2+17062071045 E_4^3 E_6^4 
\\
&\, +1415010841
   E_6^6\Bigr) \, D^2   \tvths{2}_\mu
   + (-1)^{\mu} G_{\mu-2\mu_1+2}(\tau) \, h_{1}(\tau)^2,
\end{split}
\label{genX8mod}
\ee
where $G_\mu=\frac12\, V_2[H]_\mu$ is the image of the generating series of Hurwitz class numbers \eqref{H01} 
under the Hecke-type operator $V_2$ \eqref{defHecke-exp}.
This predicts the following rational DT invariants (which turn out to be integer in this case)
\be
\begin{split}
h_{2,0} =&\, \q^{-\frac{13}{3}} \left( \underline{-10 + 2664 \q - 344564 \q^2 + 28739232 \q^3 - 2697505724 \q^4} + 228996956288 \q^5  \right. 
\\
& \left. - 3674883925530040 \q^6 +767158958736950208 \q^7
+ 7316365166650368250826 \q^8 + \dots\right),
\\
h_{2,1} =&\, \q^{-\frac{101}{24}}  \left( \underline{0+0\q+206528 q^2 - 51750016 \q^3 + 6162707008 \q^4} + 3811912230976 \q^5 \right. \\
& \left. -  1398583666309568 \q^6+5593557858350659840 \q^7
+ 15312663983997350103488 \q^8+ \dots\right),
\\
h_{2,2} =&\, \q^{-\frac{23}{6}}
\left(\underline{-17 + 6216 \q - 1041450 \q^2 + 748974840 \q^3} - 161127538287 \q^4 \right.
\\
&\left. +
 52570589013672 \q^5 + 18286285510708098 \q^6++200036575607918035704 \q^7+ \dots\right),
\end{split}
\label{exph2}
\ee
while the generating series of integer DT invariants are given by
\be
\begin{split}
\hint_{2,0} =&\, \q^{-\frac{13}{3}} \left( \underline{-10 + 2664 \q - 344564 \q^2 + 28739232 \q^3
-2697490972 \q^4} + 228996956288 \q^5 \right. \\
& \left.- 3674883927683832 \q^6 +
 767158958736950208 \q^7 + 7316365166645010675082 \q^8
+ \dots\right),
\\
\hint_{2,1} =&\, \q^{-\frac{101}{24}}  \left( \underline{0+0\q+206528 q^2 - 51750016 \q^3 + 6162707008 \q^4} + 3811912230976 \q^5 \right. \\
& \left. -  1398583666309568 \q^6+5593557858350659840 \q^7
+ 15312663983997350103488 \q^8+ \dots\right),
\\
\hint_{2,2} =&\, \q^{-\frac{23}{6}}
\left(\underline{
-16 + 6216 \q - 1041672 \q^2 + 748974840 \q^3} - 161127516752 \q^4 \right.\\
&\left. +
 52570589013672 \q^5 + 18286285477473064 \q^6 +
 200036575607918035704 \q^7+ \dots\right).
\end{split}
\ee

One can make one additional check. Comparing the coefficients of $\q$, $\q^2$ and $\q^3$ in the expansion of $h_{2,2}$ \eqref{exph2}
with predictions of \eqref{WCF-rank2-main}, one finds the following PT invariants
\be
\begin{split}
\PT(18, -66) =&\, -593492954425169689587370336688497520,
\\
\PT(18, -65) =&\, 3400776365374405985628183925346634920,
\\
\PT(18, -64) =&\,  -18783627947594252131273789893202153760.
\end{split}
\label{PT18}
\ee
These invariants in fact agree with the prediction from \cite[Thm. 1]{Alexandrov:2023zjb}, even though the condition $f(x)<\alpha$ is not satisfied (this condition is saturated in the first case, and 
so \cite[Prop. 1]{Alexandrov:2023zjb} ensures that the formula in the theorem continues to hold).
On the other hand, if one computes the very first term in the expansion of $h_{2,2}$
with spectral flow parameter $k_0+1$, one gets 
\be
\label{X8test}
\begin{split}
&\,
-698844287571739765077573369634297775026427919017 
\\
&\, 
- \tfrac{2129617568925}{2}\, \PT(18, -66) + 18732246000 \PT(18, -65) -169487850 \PT(18, -64).
\end{split}
\ee
Upon substituting \eqref{PT18}, this gives back the correct value $-17$.
This is very strong evidence that the generating series \eqref{genX8mod} is correct.

\section{Discussion}
\label{sec_disc}

In this work we obtained generating series of D4-D2-D0 indices 
(or rank 0 DT invariants) with $r=2$ units 
of D4-brane charge for two CY threefolds with one K\"ahler modulus, 
the octic $X_{8}$ and decantic $X_{10}$.
This was done by generalizing the theorem from \cite{Alexandrov:2023zjb} which expresses 
PT invariants close to the Castelnuovo bound in terms of 
rank 0 DT invariants 
and PT invariants with lower charges. The new theorem
has weaker conditions of applicability, which 
allows, on the one hand, to go further away from the Castelnuovo bound 
and, on the other hand, to get additional contributions involving D4-D2-D0 indices
with $r=2$. The resulting formula can then be used to express these indices
in terms of PT invariants or, by the GV/PT correspondence, in terms of GV invariants. The generating series are then found by computing a sufficient number of D4-D2-D0 indices and applying modularity constrains.

A new feature of the generating functions for $r=2$ compared to 
the Abelian case $r=1$ is that they are (mixed) {\it mock} modular forms.
To our knowledge, this is the first time where mock modular forms arise 
for CY threefolds without any additional structure such as an elliptic or K3 fibration. While physically mock modularity follows from the interplay of S-duality and wall-crossing in string theory compactified 
on this threefold, its mathematical origin remains mysterious.

Since the generating series found in this way provide an infinite number of new D4-D2-D0 indices, the same theorem can be used to predict infinite families of GV invariants, which provide new boundary conditions for fixing the holomorphic ambiguity 
in the direct integration method. More specifically, our results allow to push the maximal genus
of GV invariants computable by direct integration from 70 to 95 for $X_{10}$ and from 84 to 112
for $X_8$.

There are several obvious directions to try and extend our results.
First, it is natural to apply the same strategy
to find generating functions of D4-D2-D0 indices for other CY threefolds.
For example, for $r=1$ we succeeded to do this for 11 of the 13 hypergeometric models \cite{Alexandrov:2023zjb}.
The difficulty however is that this requires the knowledge of GV invariants at very high genus,
and in all cases that we analyzed, our current knowledge is not sufficient.
Thus, one has to either find a new source of boundary conditions (e.g. a better understanding of the structure of 
topological string free energies at singular points in moduli space)
or further relax the conditions of the theorem underlying our computation 
(e.g. in the way presented in \S\ref{subsec-generThm}, which was already indispensable in the case of $X_8$).

To illustrate the situation, let us consider the example of the CY threefold $X_{6,6}$
which is distinguished by the fact that the generating function $h_{2,\mu}$
has only 5 polar terms,  the smallest number among all hypergeometric models. Due to our lack of understanding of the behavior of topological string free energies at the K-point 
in K\"ahler moduli space,  the direct integration method 
allows to determine GV invariants only up to $g=21$, even with the input 
of Abelian D4-D2-D0 indices. This unfortunately only allows to 
determine the leading polar coefficient in $h_{2,0}$. It may be possible to extend Theorem \ref{thm-rank2}
so that two or three more coefficients could be found. But this is not enough yet
to fix the generating function $h_{2,\mu}$ uniquely. Thus, even this simplest case
seems to require new insights.

A second obvious direction is to try to go to higher D4-brane charge $r>2$.
In this case the generating functions $h_{r,\mu}$ are actually mock modular forms 
of depth $r-1$, which have a more complicated modular anomaly \cite{Alexandrov:2018lgp}. 
Therefore, the first necessary step is to find a solution of the modular anomaly 
equation \cite{AB-in-progress}. This would reduce the problem again to computing 
the polar terms. On one hand, this requires a generalization of Theorem \ref{thm-rank2}
with even more relaxed conditions allowing for contributions of D4-D2-D0 indices 
with $r$ units of D4-brane charge, as discussed in Remark \ref{rem_higherrk}.
On the other hand, higher D4-brane charges require GV invariants of higher and higher genus.
Fortunately, the knowledge of generating functions for $r'<r$ furnishes additional 
boundary conditions for direct integration so that one may hope that recursively 
they will be sufficient to get the required polar terms.
However, this is far from evident, as already demonstrated by the $r=2$ case for 
most of one-parameter CYs.
In fact, our estimations in \S\ref{sec_bcmock} suggest that, starting from some $r$, 
the number of new boundary condition produced by
the generating function $h_{r,\mu}$ may not be enough to push 
the direct integration even to the next genus.

Finally, a very interesting open problem is to understand the microscopic 
origin of the D4-D2-D0 indices, which we computed indirectly by studying
the chamber structure of an anti-D6-brane with suitable D2 and D0 charges in the space of weak stability conditions.
As apparent from Figures \ref{figChamberX10} and \ref{figChamberX8}, the latter is quite complicated, and the contribution from the emission of a
D4-D2-D0 brane is typically responsible for a very small fraction of the PT invariant at large volume. We expect that at least for polar charges, the chamber structure of  D4-D2-D0 should be much simpler, with a few walls arising
from emission of \DDb bound states. Unfortunately, we currently have little control on these walls, and even the first subleading polar term in $h_1$
for $X_{10}$ \eqref{resfunX10} remains puzzling from this point of view.

\appendix

\section{by S. Feyzbakhsh}
\label{sec_proofs}
As in \cite[\S A]{Alexandrov:2023zjb}, let 
$(\CY,H)$ be a smooth polarized  CY threefold with $\Pic(\CY)=\IZ.H$
satisfying the BMT conjecture \eqref{BMTineq0}. We denote by $U=\{(b,w), w>\frac12 b^2\}$
the domain of weak stability conditions
$\nu_{b,w}$ with $w=\frac12(a^2+b^2)$ (see \S\ref{sec_weak} and footnote \ref{fooabw}), and 
$\partial U=\{(b,w), w=\frac12 b^2\}$ its boundary.

\subsection{Statement of the Theorem}
\label{subsec-thm}

By studying wall-crossing for the rank $-1$ 
class 
\be
\label{defv}
\v = (-1, 0, \beta, -m),
\ee
we shall prove the following:  

\begin{Thm}\label{thm-rank2}
Fix $\beta \in H_2(\CY, \Z)$ and $m \in \Z$ such that the line $\ell_f$ of equation 
$w = -\frac{3m}{2\beta.H} b - \frac{\beta.H}{H^3}$ intersects $\partial U$ 
at two points with $b$-values $0< b_1 < b_2$ satisfying $ b_1 < 3$ and $b_2-b_1 > b_1$. 
Then 
\begin{equation}\label{WCF-rank2}
\mathrm{P}_{m, \beta} = P_1 +P_2 +P_3,
\end{equation}
where each term $P_i$ is as follows:
	
\bigskip
	
(1) The first term $P_1$ is given by
\begin{equation*}
P_1 = \sum_{(m_1, \beta_1) \in M_1}  (-1)^{\chi_{m_1, \beta_1}}\, \chi_{m_1, \beta_1} \mathrm{P}^{w(\beta_1)}_{m_1, \beta_1}\; \bOm_{\infty}\left(0,\ H, \ \frac{1}{2}\,H^2 -\beta_1 +\beta\ , \ \frac{1}{6}\,H^3 +m_1-m -\beta_1.H \right),
\end{equation*}
where	
\bea
\chi_{m_1, \beta_1} &=& m+\beta.H -\frac{H^3}{6}- \frac{1}{12}\,c_2(\CY).H -m_1+\beta_1.H ,
\\
\mathrm{P}^{w(\beta_1)}_{m_1, \beta_1} &=& \lim_{\epsilon\to 0^+} \bOm_{b=2,\, w(\beta_1)+ \epsilon}(-1, 0, \beta_1, -m_1) \quad \text{for} \quad w(\beta_1) = \frac{2\beta.H}{H^3} - \frac{3\beta_1.H}{H^3} - 1.   \label{defptw}
\eea
Moreover $M_1$ consists of pairs $(m_1, \beta_1)$ such that 
\be
\begin{split}
\label{aa-13}
0 &\,\leq\, \frac{\beta_1.H}{H^3} \leq\; \frac{3m}{2\beta.H} + \frac{\beta.H}{H^3} + \frac{1}{2}\, ,
\\
-	w(\beta_1)\,\frac{\beta_1.H}{3} - \frac{(\beta_1.H)^2}{3H^3} 
& \leq \; \ m_1 \; \ \leq\;  \frac{1}{2H^3}\, (\beta.H -\beta_1.H)^2 + \frac{1}{2}\,(\beta.H +\beta_1.H) +m.
\end{split}
\ee

\bigskip

(2) The second term $P_2$ is given by 
\begin{align*}
P_2 = \sum_{\substack{(m_1,m_0, \beta_1, \beta_0) \in M_2
}}
& \ \frac{1}{2}\ (-1)^{\chi_{m_1, \beta_1} + 	\chi'_{m_1, \beta_1} +1 }\, 
\chi_{m_1, \beta_1} \, \chi'_{m_1, \beta_1}\, 
\mathrm{P}_{m_1, \beta_1} \bOm_{\infty}(0, H, \beta_0, m_0)
\\
& \times \bOm_{\infty}\left(0,\ H,\ \beta -\beta_0 -\beta_1 +2H^2, \ -m-m_0 +m_1-2\beta_1.H + \frac{4}{3}\,H^3  \right),
\end{align*}
where 
\be
\begin{split}
\chi_{m_1, \beta_1} &= m_0-2\beta_0.H-\beta_1.H+2H^3 + \frac{1}{12}\,c_2(\CY).H ,
\\
\chi'_{m_1, \beta_1} &= -m-m_0+m_1-2\beta.H +2\beta_0.H-\beta_1.H -\frac{2}{3}\,H^3 + \frac{1}{12}\,c_2(\CY).H
\end{split}
\ee
and $M_2$ consists of quadruples $(m_1,m_0, \beta_1, \beta_0)$ such that 
\be
\begin{split}
\label{cc-125}
0 \leq\; \frac{\beta_1.H}{H^3}\; \leq\; 2 + \frac{3m}{ \beta.H} + \frac{\beta.H}{H^3} \, , 
&\qquad 
\frac{\beta_0.H}{H^3} = 1 + \frac{1}{2H^3}(\beta.H -\beta_1.H),
\\
-\frac{(\beta_1.H)^2}{2H^3} - \frac{\beta_1.H}{2}
\leq \; m_1, 
&\qquad  
m_0\; \leq\; \frac{(\beta_0.H)^2}{2H^3} + \frac{H^3}{24}\, , 
\\
-m-m_0 +m_1-2\beta_1.H + \frac{4}{3}\,H^3\; & \leq\; \frac{(\beta_0.H)^2}{2H^3}  + \frac{H^3}{24}\,. 
\end{split}
\ee
		
\bigskip 
		
(3) The third term $P_3$ is given by 

\begin{equation*}
P_3 = \sum_{(m_1, \beta_1) \in M_3}  (-1)^{\chi_{m_1, \beta_1}}\, 
\chi_{m_1, \beta_1}\, \mathrm{P}_{m_1, \beta_1}\, 
\bOm_{\infty}\left(0,\ 2H,\ \beta -\beta_1 +2H^2, \ -m +m_1 -2\beta_1.H + \frac{4}{3}\, H^3 \right),
\end{equation*}
where 
\begin{equation}
\chi_{m_1, \beta_1} = m-m_1+2\beta_1.H +2\beta.H -\frac{4}{3}\,H^3 - \frac{1}{6}\,c_2(\CY).H
\end{equation}
and $M_3$ consists of pairs $(m_1, \beta_1)$ such that 
\be
\begin{split}
\label{bb-12}
0 & \leq\, \frac{\beta_1.H}{H^3}  \leq\; 2 + \frac{3m}{ \beta.H} + \frac{\beta.H}{H^3}
\\
-\frac{(\beta_1.H)^2}{2H^3} - \frac{\beta_1.H}{2}
& \leq \;\ m_1\; \ \leq\; \frac{1}{4H^3}\,(\beta.H-\beta_1.H +2H^3)^2 - H^3 +m +2\beta_1.H. 
\end{split}
\ee
		
\end{Thm}

Furthermore, the invariant $\mathrm{P}^w_{m_1, \beta_1}$ 
defined in \eqref{defptw} can be computed by the following: 
\begin{Prop}\label{lem.rank -1}
 Let $\ell_{\v_1}(w)$ be the line passing through $\Pi(\v_1) = \left(0, -\frac{\beta_1.H}{H^3}\right)$ and $(b=2, w)$.
If $\ell_{\v_1}(w)$ intersects $\partial U$ at two points with $b$-values $b_1 <b_2$ such that $b_2 -b_1 > \max\{b_1 , 1\}$ and $0 < b_1 < 2$, then
\begin{align*}
\mathrm{P}^w_{m_1, \beta_1} = &\ \mathrm{P}_{m_1, \beta_1}\ -
\\
& \sum_{(m', \beta') \,\in\, M^w} (-1)^{\chi_{m', \beta'}}\, \chi_{m', \beta'}\, 
\mathrm{P}_{m', \beta'}\; \bOm_{\infty}\left(0,\ H, \ \frac{1}{2}\,H^2 -\beta' +\beta_1\ , \ \frac{1}{6}\,H^3 +m'-m_1 -\beta'.H \right),
\end{align*} 	
where 
\begin{equation}
\chi_{m', \beta'} = m_1+\beta_1.H -\frac{1}{6}\, H^3- \frac{1}{12}\,c_2(\CY).H -m'+\beta'.H, 
\end{equation}
and the set $M^w$ consists of pairs $(m', \beta')$ such that
\be
\begin{split}
\label{T1}
0 & \leq \beta'.H  <\frac{H^3}{2} (1-w) + \frac{\beta_1.H}{2}
\\
-\frac{(\beta'.H)^2}{2H^3}  - \frac{\beta'.H}{2}
& \leq \ m' \ \leq  \frac{1}{2H^3}\, (\beta_1.H -\beta'.H)^2 + \frac{1}{2}\,(\beta_1.H +\beta'.H) +m_1. 
\end{split}
\ee
\end{Prop}
\begin{proof}[Proof of Proposition \ref{lem.rank -1}]
One only need to apply the wall-crossing formulae of \cite{Joyce:2008pc} for the class $\v_1 = (-1, 0, \beta_1, -m_1)$ along all possible walls above the line $\ell_{\v_1}(w)$. The only difference to  \cite[Thm 1]{Alexandrov:2023zjb} is that the projection $\Pi(v')$ of the destabilising factor $v' = (-1, 0, \beta', -m') \otimes \cO_{\CY}(H)$  
must lie above the line $\ell_{\v_1}(w)$, i.e. 
\begin{equation}
\frac{w}{2} - \frac{\beta_1.H}{2H^3} <
\frac{1}{2} -\frac{\beta'.H}{H^3}\, ,
\end{equation}
as claimed in \eqref{T1}. 
\end{proof}

\bigskip 

\subsection{Proof of Theorem \ref{thm-rank2}}

As in \cite[\S A]{Alexandrov:2023zjb}, the proof proceeds by studying wall-crossing for the rank $-1$ class $\v = (-1, 0, \beta, -m)$ in the space of weak stability conditions
$\nu_{b,w}$. As before, we denote by $\bOm_{b,w}(\v)$ the (rational) Donaldson-Thomas
invariant associated to the moduli stack of $\nu_{b,w}$-semistable 
objects of class $\v$. 

In the domain $U$ parametrized by $(b,w)$, the BMT inequality \eqref{BMTineq0} implies the linear inequality
\begin{equation}\label{BMTineq}
L_{b,w}(\v) \coloneqq w(2H^3\beta.H) + 3b(-H^3)(-m) +2(\beta.H)^2 \geq 0\,.
\end{equation}
Thus, all walls for class $\v$ lie above or on the line $\ell_f$ of equation 
\begin{equation}\label{final line}
	w = -\frac{3m}{2\beta.H}\, b - \frac{\beta.H}{H^3}\, .  
\end{equation}	
Let $b_1< b_2$ be the values of $b$ at the intersection points of the line $\ell_f$ with the boundary $\partial U=\{(b,w), w = \frac{b^2}{2}\}$. 
Let $E' \rightarrow E \rightarrow E''$ be a destabilising sequence along a wall $\ell$ for an object $E$ of class $\v$. Applying the same argument as in Lemma \cite[Lemma 1]{Alexandrov:2023zjb} implies the following:  
\begin{lemma}\label{lem.destabilising objects}
Take an integer $r \geq 1$. Suppose
\begin{enumerate}
\item[(i)] $0 < b_1 < r+1$,
\item[(ii)] $b_1 < b_2-b_1$.
\end{enumerate}
 Then $b_1 \geq 1$ and there is an ordering $E_0,E_1$ of $E',E''$ such that
\begin{itemize}
\item 
$E_0$ is a rank zero slope-semistable sheaf with $0 < \ch_1(E_0).H^2 \leq r H^3$,   

\item 
$E_1$ is a rank $-1$ object such that $\cH^{-1}(E_1)$ is of rank one. 
\end{itemize}
Moreover, there is no $\nu_{b,w}$-semistable object of class $\v$ for $(b,w) \in U$ below $\ell_f$. 
\end{lemma}
\begin{proof}
 As in the proof of \cite[Lemma 1]{Alexandrov:2023zjb}, write $\ch_{\leq 1}(E_0) = (r', c'H)$ and $\ch_{\leq 1}(E_1) = (-1-r', -c'H)$. Then moving along the wall $\ell$ gives 
\begin{equation}\label{A1}
b_2 r' \leq c \leq b_1(1+r').
\end{equation}
Thus $r'(b_2-b_1) \leq b_1$ and so $r'=0$ by assumption (ii). 
Therefore, we get $c' \leq b_1 < b_2-b_1$. 
Then assumption (i) gives $c\leq r$ and the same argument as in the proof of \cite[Lemma 1]{Alexandrov:2023zjb} implies that there is no wall for $E_0$ when we move up from the wall $\ell$ to the large volume limit. Hence $E_0$ is $\nu_{b,w}$-semistable for $w \gg 1$, so is a slope-semistable sheaf by \cite[Lemma 2.7(c)]{bayer2016space}.
	
Finally, set $r_1:=\rk\!\(\cH^{-1}(E_1)\)\ge0$. By definition of the heart $\cA_b$ and using the fact that $\ell$ lies above or on $\ell_f$, we get
\begin{equation}
\ch_1\(\cH^{-1}(E_1)\).H^2 \le b_1\, r_1H^3 \qquad \text{and} \qquad  \ch_1\(\cH^0(E_1)\).H^2 \ge b_2\,(r_1 -1)H^3. 
\end{equation}
Then if $r_1 \geq 2$, we get $-c' \geq b_2 -2b_1 >0$ which makes a contradiction.  
\end{proof}
Employing Lemma \ref{lem.destabilising objects} for $r=3$ shows that there are only three potential semistable factors in the wall-crossing formula for any wall for the class $\v$:
\begin{enumerate*}
\item[(1)] $\v_1 = \cO_{\CY}(H)\otimes (-1, 0, \beta_1, -m_1)$ and 
$$
\v_0 = \left( 0,\ H,\ \beta-\beta_1 + \frac{1}{2}\, H^2, \ -m+m_1-\beta_1.H + \frac{1}{6}\, H^3  \right);
$$

\item [(2)] $\v_1= \cO_{\CY}(2H)\otimes (-1, 0, \beta_1, -m_1)$, $\v_0 = (0, H, \beta_0, m_0) $ and 
$$
\v_0' = \left(0,\ H,\ \beta -\beta_0 -\beta_1 +2H^2, \ -m-m_0 +m_1-2\beta_1.H + \frac{4}{3}\,H^3  \right);
$$  

\item [(3)] $\v_1 = \cO_{\CY}(2H) \otimes (-1, 0, \beta_1, -m_1)$ and 
$$
\v_0 = \left(0,\ 2H,\ \beta -\beta_1 +2H^2, \ -m +m_1 -2\beta_1.H + \frac{4}{3}\,H^3     \right).
$$ 
\end{enumerate*} 
	
\bigskip
	
\textbf{Case (1):} In this case, the wall-crossing formula of \cite{Joyce:2008pc} gives 
\begin{equation}
\bOm_{b,w^-}(\v) = \bOm_{b,w^+}(\v) + (-1)^{\chi(\v_1, \v_0) +1} \,\chi(\v_1, \v_0) \,\bOm_{b,w^+}(\v_0)\,\bOm_{b,w^+}(\v_1). 
\end{equation}
We know that $\Pi(\v_1) = \left(1, \ -\frac{\beta_1.H}{H^3} + \frac{1}{2}  \right)$ lies above $\ell_f$ which implies 
\begin{equation}\label{a-1}
0 \leq \frac{\beta_1.H}{H^3}\ \leq\ \frac{3m}{2\beta.H} + \frac{\beta.H}{H^3} + \frac{1}{2}\, .
\end{equation}
The equation of the wall $\ell$ passing through $\Pi(\v)$ and $\Pi(\v_1)$ is 
\begin{equation}
w = b \left(\frac{\beta.H}{H^3} - \frac{\beta_1.H}{H^3} + \frac{1}{2} \right) - \frac{\beta.H}{H^3}\, .
\end{equation}
Let $(3,w_3)$ be the intersection of $\ell$ with the vertical line	$b=3$.
Since, upon tensoring with $\cO_\CY(-H)$, 
$\bOm_{b,w^+}(\v_1) = \bOm_{b-1,\ w^+-b+\frac{1}{2}}(\v_1(-H))$, we get 
\begin{equation}
\bOm_{b =3,w_3^+}(\v_1) = \mathrm{P}^{w(\beta_1)}_{m_1, \beta_1}\, ,
\end{equation}
where $w(\beta_1) \coloneqq  w_3 -3 + \frac{1}{2} $, i.e. 
\begin{equation}
w(\beta_1) = \frac{2\beta.H}{H^3} - \frac{3\beta_1.H}{H^3} - 1. 
\end{equation} 
Applying the BMT inequality \eqref{BMTineq0} for $\v_1(-H)$ at $(b=2, w(\beta_1))$ implies 
\begin{equation}\label{a-2.5}
w(\beta_1)\,\frac{2\beta_1.H}{H^3} + \frac{6m_1}{H^3} + 2 \left( \frac{\beta_1.H}{H^3}\right)^2 \geq 0\, .
\end{equation}
Moreover, we have $\bOm_{b,w^+}(\v_0) = \bOm_{\infty}(\v_0)$ by Lemma \ref{lem.destabilising objects}, so \cite[Lemma B.3]{Feyzbakhsh:2021rcv} implies    
\begin{equation}\label{a-3}
m_1 \ \leq\  \frac{1}{2H^3} (\beta.H -\beta_1.H)^2 + \frac{1}{2}(\beta.H +\beta_1.H) +m.
\end{equation}
	
\bigskip
\textbf{Case (2):} Since $\chi(\v_0, \v_0') = 0$, the wall-crossing formula of \cite{Joyce:2008pc} implies that the contribution of type (2) walls is given by
\begin{equation}
\bOm_{b,w^-}(\v) = \bOm_{b,w^+}(\v) + \frac{k}{4} (-1)^{\chi(\v_0+ \v_0',\ \v_1) } \,\chi( \v_0, \ \v_1)\,\chi( \v_0', \ \v_1) \,\bOm_{b,w^+}(\v_0') \,\bOm_{b,w^+}(\v_0) \,\bOm_{b,w^+}(\v_1)
\end{equation}
for some constant $k$. To compute $k$, we need to sum over all the following 6 possibilities in \cite[Equation (5.13)]{Joyce:2008pc} if $\v_0 \neq \v_0'$:
\begin{enumerate*}
\item $-U(\v_0, \v_1, \v_0') = -S(\v_0, \v_1, \v_0') =1$,

\item $U(\v_0, \v_0', \v_1) = \frac{1}{2}S(\v_0+\v_0', \v_1) +S(\v_0, \v_0', \v_1) = -\frac{1}{2} +1 = \frac{1}{2}$,

\item $-U(\v_0', \v_1, \v_0) = 1$,

\item $U(\v_0', \v_0, \v_1) = \frac{1}{2}S(\v_0+\v_0', \v_1) +S(\v_0', \v_0, \v_1) = \frac{1}{2}$,

\item $U(\v_1, \v_0, \v_0') = \frac{1}{2}S(\v_1, \v_0+\v_0') +S(\v_1, \v_0, \v_0') = \frac{1}{2} +0 = \frac{1}{2}$,

\item $U(\v_1, \v_0', \v_0) = \frac{1}{2}S(\v_1, \v_0+\v_0') +S(\v_1, \v_0', \v_0) = \frac{1}{2} +0 = \frac{1}{2}$,
\end{enumerate*}  
which implies $k=4$.  If $\v_0 = \v_0'$, then there are only 3 possibilities: 
\begin{enumerate*}
\item $-U(\v_0, \v_1, \v_0') = -S(\v_0, \v_1, \v_0') =1$,

\item $U(\v_0', \v_0, \v_1) = \frac{1}{2}S(\v_0+\v_0', \v_1) +S(\v_0', \v_0, \v_1) = \frac{1}{2}$,

\item $U(\v_1, \v_0, \v_0') = \frac{1}{2}S(\v_1, \v_0+\v_0') +S(\v_1, \v_0, \v_0') = \frac{1}{2} +0 = \frac{1}{2}$,
\end{enumerate*}  
which implies $k=2$. 
We know that $\v_1 = \cO_{\CY}(2H) \otimes (-1, 0, \beta_1, -m_1)$ 
and $\Pi(\v_1) = (2, 2- \frac{\beta_1.H}{H^3})$ lies above $\ell_f$, i.e. 
\begin{equation}
-\frac{3m}{\beta.H}  - \frac{\beta.H}{H^3} \leq 2- \frac{\beta_1.H}{H^3}\, ,
\end{equation}
or equivalently 
\begin{equation}\label{c-1}
0 \leq\ \frac{\beta_1.H}{H^3}\ \leq\ 2 + \frac{3m}{ \beta.H} + \frac{\beta.H}{H^3}. 
\end{equation}
Then \cite[Lemma 3.5]{feyz:effective-restriction-theorem} implies that there is no wall for class $\v_1$ crossing the vertical line $b=3$.\footnote{Since $b \in \mathbb{Z}$ and $\Pic(\CY) = \Z .H$, we have that $H^3$ divides $H^2.\ch_1 -b H^3.\ch_0  \in \Z$.} 
Thus $\bOm_{b,w^+}(\v_1) = \mathrm{P}_{m_1, \beta_1}$ and the Castelnuovo bound \cite[Thm 2]{Alexandrov:2023zjb} gives 
\begin{equation}\label{c-2}
-\frac{2(\beta_1.H)^2}{3H^3}  - \frac{\beta_1.H}{3}
\leq \ m_1. 
\end{equation}
Since $\v_0 = (0, H, \beta_0, m_0)$, the slope of the wall must be equal to $\frac{\beta_0.H}{H^3}$, i.e.
\begin{equation}\label{c-3}
\frac{\beta_0.H}{H^3} = 1 + \frac{1}{2H^3}(\beta.H -\beta_1.H)
\end{equation}
and 
\begin{equation}\label{c-4}
m_0\ \leq\ \frac{(\beta_0.H)^2}{2H^3} + \frac{H^3}{24}. 
\end{equation}
Similarly, for the rank 0 class $\v_0' = \v -\v_1-\v_0$ we get 
\begin{equation}\label{c-5}
-m-m_0 +m_1-2\beta_1.H + \frac{4}{3}\, H^3\ \leq\ \frac{(\beta_0.H)^2}{2H^3}  + \frac{H^3}{24}\,. 
\end{equation}
Note that by \eqref{c-3}, we have $(\beta -\beta_0 -\beta_1 +2H^2).H = \beta_0.H$. Finally, we have 
\begin{align*}
\chi(\v_0+ \v_0',\ \v_1) = -m +m_1 + \frac{4}{3}\,H^3 + \frac{1}{6}\,c_2(\CY).H -2\beta_1.H -2\beta.H.
\end{align*}

\bigskip
	
\textbf{Case (3):} The contribution of type (3) walls is given by 
\begin{equation}
\bOm_{b,w^-}(\v) = \bOm_{b,w^+}(\v) + (-1)^{\chi(\v_1, \v_0) +1} \,\chi(\v_1, \v_0)\, \bOm_{b,w^+}(\v_0)\,\bOm_{b,w^+}(\v_1). 
\end{equation}
By a similar argument as in Case (2), there are no walls for both $\v_0$ and $\v_1$ till the large volume limit, so $\bOm_{b,w^+}(\v_0) = \bOm_{\infty}(\v_0)$ and $\bOm_{b,w^+}(\v_1) =\bOm_{\infty}(\v_1) = \mathrm{P}_{m_1, \beta_1}$. 
We know that $\Pi(\v_1) = \left(2, 2 - \frac{\beta_1.H}{H^3} \right)$ lies above $\ell_f$, which implies \begin{equation}\label{b-1}
0 \leq\ \frac{\beta_1.H}{H^3}\ \leq\ 2 + \frac{3m}{ \beta.H} + \frac{\beta.H}{H^3}\, . 
\end{equation}
Moreover, the Castelnuovo bound \cite[Thm 2]{Alexandrov:2023zjb} and \cite[Lemma B.3]{Feyzbakhsh:2021rcv} give
\begin{equation}
\begin{split}
    \label{b-2}
-\frac{2(\beta_1.H)^2}{3H^3} - \frac{\beta_1.H}{3}
& \leq \; m_1\, , 
\\
-m +m_1 -2\beta_1.H + \frac{4}{3}\,H^3 & \leq \frac{1}{4H^3}(\beta.H-\beta_1.H +2H^3)^2   + \frac{1}{3}\, H^3. 
\end{split}
\end{equation} 
Finally, we have 
\begin{equation}
\chi(\v_1, \v_0) = m-m_1+2\beta_1.H +2\beta.H -\frac{4}{3}\,H^3 - \frac{1}{6}\,c_2(\CY).H\, ,
\end{equation}
which ends the proof.

\subsection{Remarks}

\begin{remark}
First, we rephrase Theorem \ref{thm-rank2} in terms of the variables 
\begin{equation}
x= \frac{\beta.H}{H^3} \qquad \text{and} \qquad \alpha = -\frac{3m}{2 \beta.H}\, ,
\end{equation}
as in \eqref{defxa}. These variables are chosen such that 
the equation of the BMT wall $\ell_f$ is $w = \alpha b -x$.
Assuming that $\alpha\geq \sqrt{2x}$, $\ell_f$ intersects the parabola
$w=\frac12 b^2$ at $b_1, b_2 = \alpha \pm \sqrt{\alpha^2 -2x}$. The condition $b_2-b_1 >b_1$ is equivalent to 
\begin{equation}
3\sqrt{\alpha^2-2x} > \alpha \qquad \text{i.e.} \qquad \alpha> \frac32 \sqrt{x}\, .
\end{equation}
The condition $b_1 < 3$ is equivalent to 
\begin{equation}
\alpha - \sqrt{\alpha^2-2x} < 3 \qquad \text{i.e.} \qquad \sqrt{2x}<\alpha <3 \qquad \text{or} \qquad \frac{3}{2} + \frac{x}{3} < \alpha.
\end{equation}
The conjunction of the two conditions is therefore $\alpha>\frac32 \sqrt{x}$ when $1<x<9$ and $\alpha>\frac{3}{2} + \frac{x}{3}$ when $x>9$. In the range $1<x<4$, we have the stronger condition $b_1<2$,
so \cite[Thm 1]{Alexandrov:2023zjb} applies. Thus, we can rephrase Theorem \ref{thm-rank2} as in \S\ref{sec_main}, by replacing the function 
$f(x)$ in \cite[Thm 1]{Alexandrov:2023zjb} by the function $f_2(x)$
defined in \eqref{deffx}.
\end{remark}

\begin{remark} \label{rkvk}
Consider an arbitrary rank zero class $\v = (0, 2H, cH^2, dH^3)$ with $c,d\in\IZ$. For any $k>0$, take a section $\cO_{\CY}(-kH) \to F$, then the cone is of class $(-1, (k+2)H, (c- \frac{k^2}{2})H^2, (d+ \frac{k^3}{6})H^3 )$. Tensoring it by $\cO_{\CY}((k+2)H)$ gives the class 
\begin{equation}
\v_k = \left(-1,\ 0,\ (c+2k+2)H^2, \ d+c(k+2)+k^2+4k+ \frac{8}{3}\right). 
\end{equation}
We claim that if $k$ is large enough, then 
\begin{enumerate*}
\item 
we can apply Theorem \ref{thm-rank2} to $\v_k$, i.e. the final wall $\ell_f$ intersects $\partial U$ at two points with $b$-values $2 \leq b_1 <3$ and $b_2-b_1 > b_1$, 

\item 
the difference of the right and left hand side of the last equation in \eqref{bb-12} is negative, i.e. 
\begin{equation}
\frac{1}{4H^3}\,(\beta.H-\beta_1.H +2H^3)^2   - H^3 +m +2\beta_1.H - \left( 	-\frac{2(\beta_1.H)^2}{3H^3} - \frac{\beta_1.H}{3}  \right) <0
\end{equation}	
for all non-zero possible values of $\beta_1$ satisfying the inequalities in the first line of \eqref{bb-12}.  
\end{enumerate*}
Clearly, these two conditions determine an explicit lower bound for $k$. If (1) and (2) hold, the only contribution to $P_3$ in \eqref{WCF-rank2} is from $\v_1 = [\cO_{\CY}(2H)[1]]$ and $\v_0 = \v$.   
\end{remark}

\begin{remark}\label{rem_higherrk}
By applying Lemma \ref{lem.destabilising objects}, one can show that Theorem \ref{thm-rank2} can be extended to include contributions from higher-rank D4-D2-D0 invariants. Specifically, suppose that $b_2 - b_1 > b_1$ and $b_1 \in [r, r+1)$ for $r\geq 1$. Then one can express $\mathrm{P}_{m,\beta}$ in terms of $\mathrm{P}_{m', \beta'}$ when $\beta'.H < \beta.H$ and rank 0 DT invariants $\bOm_{\infty}(0, r'H, \beta', m')$ for $1 \leq r' \leq r$.
\end{remark}

\subsection{Generalisations of Theorem \ref{thm-rank2}}
\label{subsec-generThm}

We now show that the formula \eqref{WCF-rank2} continues to hold in some cases even when the hypotheses of Theorem \ref{thm-rank2} are not satisfied, 
or that it can be corrected by including contributions of further walls.

\subsubsection{$\kappa =1$.} We first consider CY 3-folds $X$ with $\kappa =H^3 =1$.

\begin{Lem}
\label{LemmaX6}
If $H^3=1$, then Theorem \ref{thm-rank2} is still valid for 
\begin{enumerate}
\item[(i)] $\beta = 10H^2$ and $m= -32$ (corresponding to $\cO(\q^3)$ term in $h_{2,0}$ for $X_{10}$),  

\item[(ii)] $\beta=12H^2$ and $m=-44$ (corresponding to $\cO(\q^4)$ term in $h_{2,0}$ for $X_{10}$),
		
\item [(iii)] $\beta = 12 H^2$ and $m= -43$ (corresponding to $\cO(\q^5)$ term in $h_{2,0}$ for $X_{10}$),
		
\item [(iv)] $\beta = 11H^2$ and $m =-37$ (corresponding to $\cO(\q^4)$ term in $h_{2,1}$ for $X_{10}$).
\end{enumerate}
\end{Lem}
\begin{proof}
	One can easily check that Lemma \ref{lem.destabilising objects} is valid in all the above cases. Thus there are only three possibilities for the Chern character of $E_0$ and $E_1$: 
	\begin{enumerate}
		\item[(a)] $\ch_{\leq 2}(E_1) = (-1, -H, -\frac{1}{2}H^2 + x_1H^2)$ and $\ch_{\leq 1}(E_0) = (0, H)$ for some integer $0 \leq x_1$.
		\item[(b)]  $\ch_{\leq 2}(E_1) = (-1, -2H, -2H^2 + x_2H^2 )$ for some integer $0 \leq x_2$ 
		and $\ch_{\leq 1}(E_0) =(0, 2H)$.  
		\item[(c)]  $\ch(E_1)= (-1, -3H, -\frac{9}{2}H^2 , -\frac{9}{2}H^3 +x_3 )$ for some integer $x_3 \leq 0$ and $\ch_{\leq 1}(E_0) = (0, 3H)$.
	\end{enumerate}
First consider case (i). In  option (a), the wall lies above or on $w = \frac{11}{2} b -10$, so the conditions of Theorem \ref{thm-rank2} are satisfied. In (b), the wall lies above or on the line $w = 5b -10$ which intersects the vertical line $b=3$ at a point inside $U$, so the rank $-1$ object is stable in the large volume limit as required in the proof of Theorem \ref{thm-rank2}. Finally, we claim that option (c) cannot happen because applying \cite[Lemma B.3]{Feyzbakhsh:2021rcv} 
(or equivalently \eqref{qmax}) for $E_0$ implies that 	
\begin{equation}
\ch_3(E_0) = 32 +\frac{9}{2} -x_3 \  \leq \  \frac{(14.5)^2}{6} + \frac{27}{24}\, ,
\end{equation}
which gives $x_3 > 0$, a contradiction.

The other cases (ii), (iii) and (iv) follow via the same argument as in case (i). 
Note that in case (iii), the lowest wall where option (a) or (b) can happen is of equation $w = \frac{11}{2} b -12$ passing through $(3,\frac92)$, so the proof of Theorem \ref{thm-rank2} still works. Also in case (iv), the lowest wall for (a) and (b) is of equation $w = \frac{11}{2} b -11$, so again the proof of Theorem \ref{thm-rank2} is valid.   
	
\end{proof}

\subsubsection{$\kappa =2$.} In this part, we consider CY 3-fold $\CY= X_8$ with $\kappa =H^3 =2$. Then we can apply Koseki's improvement of the classical Bogomolov inequality \cite{koseki2022stability}. In other words, we may replace the parabola $w = \frac{b^2}{2}$ with the
curve $w=\Gamma(b)$ where the function $\Gamma \colon \mathbb{R} \to \mathbb{R}$ satisfies 
\be
\Gamma(b)\ \coloneqq \ \left\{\!\!\!\begin{array}{cc} 
b^2-b, & b \in [0, \frac{1}{4}), 
\vspace{.2 cm}
\\
\frac{3}{4}b-\frac{3}{8}\, , & b \in [\frac{1}{4}, \frac{1}{2}), 
\vspace{.2 cm}
\\
\frac{1}{4}b -\frac{1}{8}\, , & b \in [\frac{1}{2}, \frac{3}{4}), 
\vspace{.2 cm}
\\
b^2-\frac{1}{2}\, , & b \in [\frac{3}{4}, 1),
\end{array}\right.
\ee
and 
\begin{equation}
\Gamma(b) = nb -\frac{n^2}{2} + \Gamma(b-n), \qquad b \in [n, n+1).
\end{equation}
We denote by $b_1' <b_2'$ the $b$-values of the intersection points of the final line $\ell_f$ \eqref{final line} with the curve $w=\Gamma(b)$.    

\begin{Lem}
\label{LemmaX81}
	If $\CY=X_8$, then Theorem \ref{thm-rank2} is still valid for 
	\begin{enumerate}
		\item[(i)] $\beta = 8H^2$ and $m =-45$, corresponding to $\cO(\q^3)$ term in $h_{2,0}$,
		\item[(ii)] $\beta = 7H^2$ and $m= -37$, corresponding to $\cO(\q)$ term in $h_{2,2}$,
		\item[(iii)] $\beta = \frac{17}{2}H^2$ and $m = -49$, corresponding to $\cO(\q^4)$ term in $h_{2,1}$.
	\end{enumerate}
\end{Lem}
\begin{proof}
	In case (i), we have $b_1'= \frac{172}{63}\sim 2.73 < b_2'= \frac{188}{33}\sim 5.69$, and in case (ii) we have $b_1'= \frac{245}{96} < 
	2.55 < b_2'= 5.39$, so the conditions of Theorem \ref{thm-rank2} are satisfied. But in case (iii), we have $b_1' = 3.01 < b_2' = 5.77$. Then using the same notations as in Lemma \ref{lem.destabilising objects}, we get 
	\begin{equation}
		5.77 r \leq c \leq 3.01 (1+r). 
	\end{equation}
	Then there are 4 possibilities: 
\begin{enumerate}	
	\item[(a)] If $(r,c)=(1,6)$, then $\ch_{2}(E_0) = 18 H^2$ or $\frac{35}{2}H^2$. In the first case, we get $\ch_{\leq 2}(E_1) = (-2, -6H, -\frac{19}{2}H^2)$ which is not possible by \cite[Theorem 1.1]{koseki2022stability}. 
	In the second,
	$\ch(E_0) = (1, 6H, \frac{35}{2} H^2, mH^3)$ and $\ch(E_1) = \(-2, -6H, -9H^2, (-m + \frac{49}{2})H^3\)$. Since the wall intersects the vertical line $b=5$ at a point above the curve $\Gamma$, by \cite[Lemma 3.5]{feyz:effective-restriction-theorem} there is no wall for $E_0$ up to the large volume limit. Then applying the Castelnuovo
	bound \eqref{CastPT} for $E_0(-6H)$ implies that 
	\begin{equation}\label{cond.1}
	\frac{\ch_3(E_0(-6H))}{H^3} = -33+m \leq \frac{3}{8}\, . 
	\end{equation}
	Moreover, the BMT inequality \eqref{BMTineq0} for $E_1$ at the point $(b = 3+ \epsilon,w)$ along the wall implies that $m-\frac{49}{2} \geq 9$,
	which is not possible by \eqref{cond.1}. 
	
	\item[(b)] If $(r,c)=(0,3)$, then $\ch_{\leq 2}(E_1) = (-1, -3H, -\frac{9}{2}H^2)$. Thus $E_1$ is the derived dual of a stable pair, so we must have $\ch_3(E_1) = -\frac{9}{2}H^3$ and $\ch(E_0) = (0, 3H, \frac{26}{2}H^2, \frac{58}{2}H^3)$, which is not possible by Lemma \ref{lem-general}. 

	\item[(c)] If $(r,c)=(0, 1)$, then the lowest wall is of equation  $w= 5b - \frac{17}{2}$ and the proof of Theorem \ref{thm-rank2} is valid. 
        
 	\item[(d)] If $(r, c)=(0, 2)$  then the lowest wall is of equation $w= \frac{19}{4} b - \frac{17}{2}$ and the proof of Theorem \ref{thm-rank2} is again valid. 
 \end{enumerate}

\end{proof}

\begin{Lem}
\label{LemmaX82}
	If $\CY= X_8$ and $\mathrm{P}_{m=0,\, \beta = \frac{1}{2}H^2} = 0$ \footnote{By the GV/PT correspondence, 
	one has $\mathrm{P}_{m=0,\, \beta = \frac{1}{2}H^2}=\GV_1^{(1)}$. The corresponding entry in \cite[Table 5]{Huang:2006hq} vanishes. We expect
	that this vanishing can be shown rigorously using arguments similar to \cite{li2009genus}. },
	then  Theorem \ref{thm-rank2} continues to hold for $\beta = \frac{17}{2}H^2$ and $m = -48$, corresponding to $\cO(\q^5)$ in $h_{2,1}$. 
\end{Lem}
\begin{proof}
We have $b_1' \sim 3.14 < b_2' \sim 5.36$. Then using again the same notations as in Lemma \ref{lem.destabilising objects}, we get 
\begin{equation}
5.36\, r \leq c \leq 3.14\, (1+r). 
\end{equation}
	
\textbf{Case (a)} If $r=1$ and $c=6$, then we have the following cases: 
\begin{enumerate}
		\item [(a1)] $\ch_{\leq 2}(E_0) = (1, 6H, 18 H^2)$ and $\ch_{\leq 2}(E_1) = (-2, -6H, -\frac{19}{2}H^2)$. This is not possible by \cite[Theorem 1.1]{koseki2022stability}. 
		\vspace{.2 cm}
		\item [(a2)] $\ch(E_0) = (1, 6H, \frac{35}{2} H^2, (33-m)H^3)$ and $\ch(E_1) = (-2, -6H, -9H^2, (-9+m)H^3)$. Applying the BMT inequality \eqref{BMTineq0} for $E_1$ at the point $(b = 3+ \epsilon,w)$ 
		with $0<\epsilon\ll 1$ shows that $m\leq 0$. Since $\ch(E_0(-6H)) = (1, 0, -\frac{1}{2}H^2, -mH^3)$, the Castelnuovo
	bound \eqref{CastPT} and the vanishing $\mathrm{P}_{m=0,\, \beta = \frac{1}{2}H^2} = 0$ imply $m>0$, a contradiction.   
		\vspace{.2 cm}

		\item [(a3)] $\ch(E_0) = (1, 6H, 17H^2, mH^3) = (1, 0, -H^2, (m-30)H^3)\otimes \cO_{\CY}(6H)$ and \\
		$\ch(E_1) = (-2, -6H, -\frac{17}{2}H^2, (-m+24)H^3)$. Moreover, the equation of the wall is $w = \frac{51}{12} b - \frac{17}{2}$. By the Castelnuovo bound \eqref{CastPT}, we get $m-30  \leq 1$. Thus the final wall for $E_1$ given by the BMT inequality \eqref{BMTineq0} lies above or on the line with equation $w = \frac{9}{2} b -\frac{37}{4}$, which lies above our wall. This means there is no semistable object along the wall of class $\ch(E_2)$, a contradiction.  
\end{enumerate}

\bigskip 
	
\textbf{Case (b)} If $r=0$ and $c=3$, then $\ch_{\leq 2}(E_1) = (-1, -3H, -\frac{9}{2}H^2)$ and $\ch_{\leq 2}(E_0) = (0, 3H, \frac{26}{2}H^2)$. As $E_1$ is stable in the large volume by \cite[Corollary 3.11]{bayer2016space}, we must have $E_1 \cong \cO_{\CY}(3H)[1]$, and so $\ch_3(E_0) = 57$ which is not possible by Lemma \ref{lem-general}.

\bigskip
	
\textbf{Case (c)} If $r=0$ and $c=2$, then $\ch(E_1) = (-1, -2H, mH^2)$ for $m \leq 0$. If $m < 0$, the proof of Theorem \ref{thm-rank2} is valid. If $m=0$, then the wall is of equation $w= \frac{17}{4} b -\frac{17}{2}$. Thus $\ch(E_1) = (-1, 0, 2H^2, sH^3) \otimes \cO_{\CY}(2H)$ and $\ch(E_0) = (0, 2H, \frac{17}{2}H^2, (\frac{64}{3} -s )H^3)$. We know that $E_0$ is slope-semistable, so 
\begin{equation}
\frac{128}{3} -2 s \leq \frac{17^2}{8} +\frac{8}{12}\, ,
\end{equation}  
which implies $s \geq 3$. But Lemma \ref{lem-general} shows that $s \neq 3$ which is not possible by the Castelnuovo bound \eqref{CastPT} for $E_1$. 

\bigskip
	
\textbf{Case (d)} If $r=0$ and $c=1$, then $\ch(E_1) = (-1, -H, mH^2)$ for $m \leq 4$, so the wall lies above or on the line of equation $w= 9/2 b -17/2$ where the proof of Theorem \ref{thm-rank2} is valid. 
\end{proof}

\begin{Lem}\label{lem-general}
	If $\CY= X_8$, there is no $\nu_{b,w}$-semistable object $E$ of class 
	\begin{enumerate}
	    \item [(i)] $\ch(E) = (0, 3H, 13 H^2, sH^3)$ with $s \geq \frac{57}{2}$, or 
	\item [(ii)] $\ch(E) = (0, 2H, \frac{17}{2}H^2, \frac{55}{3}H^3)$. 
	\end{enumerate}
\end{Lem}
\begin{proof}
    First, assume that there is such an object as in case (i). Then the final wall given by BMT inequality \eqref{BMTineq0} lies above or on the line $\ell_f$ with equation $w = \frac{13}{3}b -\frac{163}{18}$. This line intersects the curve $\Gamma$ at two points with $b$-values $b_1' \sim 3.31 < b_2' \sim 5.34$. Let the destabilising factors $E_1$ and $E_{-1}$ be of classes $\ch_{\leq 1}(E_1) = (r, cH)$ and $\ch_{\leq 1}(E_{-1}) = (-r, (-c+3)H)$. Then the same argument as in Lemma \ref{lem.destabilising objects} gives 
    \begin{equation*}
        5.34 r \leq c \leq 3.31 r+3,
    \end{equation*}
    which implies $(r, c) = (1, 6)$. Let $\ch_2(E_1) = (18-m)H^2$, then $\ch_2(E_{-1}) = (-5+m)H^2$. Since the wall lies above or on $\ell_f$, we must have $\ch_{\leq 2}(E_1) = (1, 6H, 17H^2)$ and so $\ch_{\leq 2}(E_{-1}) = (-1, -3H, -4H^2)$. Then applying the Castelnuov bound \eqref{CastPT} to $E_1$ and $E_{-1}$ leads to a contradiction. 
    
    In case (ii), the final wall is of equation $w = \frac{17}{4}b - \frac{69}{8}$ which intersects the curve $\Gamma$ at $b_1' \sim 3.1 < b_2' \sim 5.3$. Then as in case (i), let the destabilising classes be $\ch_{\leq 1}(E_1) = (r, cH)$ and $\ch_{\leq 1}(E_{-1}) = (-r, (-c+2)H)$. Consequently, we must have $5.3 r \leq c \leq 3.1 r+2$ which is not possible. 
\end{proof}

\begin{Lem}
\label{LemmaX83}
	If $\CY=X_8$, $\beta = 8H^2$ and $m =-44$, 
	corresponding to $\cO(q^4)$ term in $h_{2,0}$,
	then the only additional wall $\ell$, compared to Theorem \ref{thm-rank2}, is formed by the classes $\v_1= (1, 0, -H^2, H^3) \otimes \cO_{\CY}(6H)$ and $\v_2=\v_3 =[\cO_{\CY}(3H)[1]]$.  Moreover, if the point $(b_0,w_0)$ lies just above the wall $\ell$, then 
	$\bOm_{b_0,w_0}(\v_1) = \DT(\kappa,-\kappa)$ and   
	\begin{equation}\label{c.1}
		\bOm_{b_0,w_0}([\cO_{\CY}(3H)]-\v_1) = (-1)^{\chi(\v_1, [\cO_{\CY}(3H)]) -1} \chi(\v_1, [\cO_{\CY}(3H)]) \DT(\kappa,-\kappa). 
	\end{equation}
\end{Lem}
\begin{proof}
	In this case, we know $b_1' \sim 3.06 < b_2' \sim 5.31$, so using the same notations as in Lemma \ref{lem.destabilising objects},
	\begin{equation}
		5.3 r \leq c \leq 3.06 (1+r). 
	\end{equation}
	
	\textbf{Case (a1).} If $r=0$ and $c=3$, then we must have $\ch(E_1) = (-1, \ -3H, \ -\frac{9}{2}H^2, \ -\frac{9}{2}H^3)$ and $\ch(E_0) = (0,\ 3H, \  \frac{25}{2}H^2, \  \frac{53}{2}H^3  )$. But applying  \cite[Lemma B.3]{Feyzbakhsh:2021rcv} for $E_0$ leads to a contradiction.     
	
	\textbf{Case (a2)} If $r=0$ and $c=2$, then $\ch(E_1) = (-1, -2H, \frac{x}{2}H^2)$ where $x \in \Z$ and $-4 \leq x \leq -1$. Consequently, the wall lies above or on the line with the equation $w= \frac{17}{4}b-8$. One can easily verify that the proof of Theorem \ref{thm-rank2} remains valid.  
	
	\textbf{Case (a3)}  If $r=0$ and $c=1$, then $\ch(E_1) = (-1, -H, \frac{x}{2}H^2)$ where $x \in \Z$ and $-1 \leq x \leq 7$. Then the wall lies above or on the line with equation $w=\frac92 b-8$. Again, the proof of Theorem \ref{thm-rank2} holds true. 
	
	\textbf{Case (b).} If $r=1$, then $c= 6$, so $\ch_{\leq 1}(E_0) = (1, 6H)$ and $\ch_{\leq 1}(E_{1}) = (-2, -6H)$.  The location of the final wall forces $\ch_2(E_1) = -9H^2$ and $\ch_2(E_0) =  17H^2$. Applying the BMT inequality \eqref{BMTineq0} for $E_1$ at the point $(b = 3+ \epsilon,w)$ on the wall, where $0 < \epsilon \ll 1$, implies that $\ch_3(E_1) \leq -9H^3$. Let $(b_0, w_0)$ be a point just above this wall. Since the wall intersects the vertical line $b=6$, we get $\bOm_{b_0, w_0}(\ch(E_0)) = \bOm_{b_0, \infty}(\ch(E_0))$. Thus the Castelnuovo bound \eqref{CastPT} implies that $\ch_3(E_1) = -9H^3$, so $\ch(E_1) = \ch(\cO_{\CY}(3H)^{\oplus 2}[1])$ and $\ch(E_0) = (1, \ 0 ,\ - H^2,\  H^3) \otimes \cO_{\CY}(6H) =\v_1 $ as claimed in Lemma. The final claim \eqref{c.1} follows from \cite[Theorem 1.1]{Feyzbakhsh:2022ydn}. 
\end{proof}

\begin{Lem}
\label{LemmaX84}
	If $\CY= X_8$, $\beta = 7H^2$ and $m =-36$, corresponding to $\cO(q^2)$ term in $h_{2,2}$, 
	then  the only additional wall $\ell$, compared to Theorem \ref{thm-rank2}, is made by the classes $[ \cO_{\CY}(5H)]$ and $v = (-2, -5H, -\frac{11}{2}H^2, -\frac{17}{6}H^3)$.  Moreover, if the point $(b_0,w_0)$ with $b_0> \frac{5}{2}$ lies just above the wall $\ell$, then $\bOm_{b_0, w_0}(v) = \bOm_{b_0, w \to \infty}(v)$. 
	\newline
	When we move from the large volume limit for class $v$ to the empty region (given by the BMT inequality \eqref{BMTineq}), we encounter only one wall for class $v$, formed by the classes $\v_1=\v_2= \v_3 = [\cO_{\CY}(3H)[1]]$ and $\v_4= [ \cO_{\CY}(4H)]$. Therefore one can compute $\bOm_{b_0, w \to \infty}(v)$ using the wall-crossing formula of \cite{Joyce:2008pc}. 
\end{Lem}
\begin{proof}
	We have $b_1' = 2.72 < b_2' = 4.89$ and $4.8 r \leq c \leq 2.7 (1+r)$. Thus either (a) $r=0$ and $c=1$ or $2$, or (b) $r=1$ and $c=5$. One can easily check that the proof of Theorem \ref{thm-rank2} is valid for case (a), so we only need to consider case (b). 
	
	When $(r, c) = (1, 5)$, the location of the final wall dictates that $\ch_{\leq 2}(E_0) = (1, 5H, \frac{25}{2}H^2)$ and $\ch_{\leq 2}(E_1) = (-2, -5H, -\frac{11}{2}H^2)$. Since there is no wall for $E_0$ crossing the vertical line $b=4$, $E_0$ is slope-stable and so $\ch_3(E_0) \leq \frac{125}{6}H^3$. On the other hand, there is no wall for $E_1$ crossing the vertical line $b=3$, thus applying the BMT inequality \eqref{BMTineq} at the point $(b, w) =(3, \frac{9}{2}+\epsilon)$ implies that $\ch_3(E_1) \leq -\frac{17}{6}$. Given that $\ch_3(E_0) + \ch_3(E_1) = 18H^2$, we conclude that $E_0 = \cO_{\CY}(5H)$ and $\ch(E_1) = v = (-2, -5H, -\frac{11}{2}H^2, - \frac{17}{6}H^3)$, completing the proof of the first part of the Lemma. 
	
	Now, we perform wall-crossing for the rank $-2$ class $v$. The final wall $\ell_{v}:w= \frac72 b-6$ for class $v$, as determined by the BMT inequality \eqref{BMTineq}, intersects the boundary $w = \frac{b^2}{2}$ at two points with $b$-values $3$ and $4$. Since there is no wall for class $v$ crossing the vertical line $b=3$, $\ell_v$ is the only wall that affects class $v$. Let $F_1$ and $F_2$ be the destabilizing objects along the wall for an object $F$ of class $v$. Writing $\ch_{\leq 1}(F_1) = (r', c'H)$, we find that $\ch_{\leq 1}(F_2) =(-r'-2, (-c'-5)H)$. For any $b \in (3, 4)$, we must have
	\begin{equation*}
	   br' \leq c' \leq -5 +b(r'+2). 
	\end{equation*}
Then, up to reordering the factors, the only possibilities for the pair $(r', c')$ are: (i) $(-1, -3)$, (ii) $(0, 1)$, and (iii) $(1, 4)$. In case (i), the location of the wall compels $\frac{1}{H^3}\ch_2(F_1).H = -\frac{9}{2}$. Consequently, $F_1 = \cO_{\CY}(3H)[1]$, and $\ch(F_2) = (-1, -2H, -H^2, \frac{5}{3}H^3) = (-1, 0, H^2, H^3) \otimes \cO_{\CY}(2H)$. Combining \cite[Theorem 1.1]{Feyzbakhsh:2022ydn} with \cite[Lemma 1]{Alexandrov:2023zjb} demonstrates that $\ell_v$ is the sole wall for $F_2$ with destabilizing factors of classes $[\cO_{\CY}(3H)^{\oplus 2}[1]]$ and $[\cO_{\CY}(4H)]$. Applying the same argument for cases (ii) and (iii) reveals that $\cO_{\CY}(4H)$ and three copies of $[\cO_{\CY}(3H)[1]]$ are all possible stable factors along the wall $\ell_v$ for class $v$.         
	
\end{proof}

\providecommand{\href}[2]{#2}\begingroup\raggedright\endgroup


\begin{thebibliography}{10}

\bibitem{Maldacena:1997de}
J.~M. Maldacena, A.~Strominger, and E.~Witten, ``{B}lack hole entropy in
  {M}-theory,'' {\em JHEP} {\bf 12} (1997) 002,
\href{http://www.arXiv.org/abs/hep-th/9711053}{{\tt hep-th/9711053}}.

\bibitem{deBoer:2006vg}
J.~de~Boer, M.~C.~N. Cheng, R.~Dijkgraaf, J.~Manschot, and E.~Verlinde, ``{A
  farey tail for attractor black holes},'' {\em JHEP} {\bf 11} (2006) 024,
\href{http://www.arXiv.org/abs/hep-th/0608059}{{\tt hep-th/0608059}}.

\bibitem{Gaiotto:2006wm}
D.~Gaiotto, A.~Strominger, and X.~Yin, ``{The M5-brane elliptic genus:
  Modularity and BPS states},'' {\em JHEP} {\bf 08} (2007) 070,
\href{http://www.arXiv.org/abs/hep-th/0607010}{{\tt hep-th/0607010}}.

\bibitem{Gaiotto:2007cd}
D.~Gaiotto and X.~Yin, ``{Examples of M5-Brane Elliptic Genera},'' {\em JHEP}
  {\bf 11} (2007) 004,
\href{http://www.arXiv.org/abs/hep-th/0702012}{{\tt hep-th/0702012}}.

\bibitem{Manschot:2007ha}
J.~Manschot and G.~W. Moore, ``{A Modern Fareytail},'' {\em Commun. Num. Theor.
  Phys.} {\bf 4} (2010) 103--159,
\href{http://www.arXiv.org/abs/0712.0573}{{\tt 0712.0573}}.

\bibitem{Klemm:2012sx}
A.~Klemm, J.~Manschot, and T.~Wotschke, ``{Quantum geometry of elliptic
  Calabi-Yau manifolds},'' \href{http://www.arXiv.org/abs/1205.1795}{{\tt
  1205.1795}}.

\bibitem{Oberdieck:2016nvt}
G.~Oberdieck and J.~Shen, ``{Curve counting on elliptic Calabi\textendash{}Yau
  threefolds via derived categories},'' {\em J. Eur. Math. Soc.} {\bf 22}
  (2019), no.~3, 967--1002, \href{http://www.arXiv.org/abs/1608.07073}{{\tt
  1608.07073}}.

\bibitem{Bouchard:2016lfg}
V.~Bouchard, T.~Creutzig, D.-E. Diaconescu, C.~Doran, C.~Quigley, and
  A.~Sheshmani, ``{Vertical D4\textendash{}D2\textendash{}D0 Bound States on K3
  Fibrations and Modularity},'' {\em Commun. Math. Phys.} {\bf 350} (2017),
  no.~3, 1069--1121, \href{http://www.arXiv.org/abs/1601.04030}{{\tt
  1601.04030}}.

\bibitem{Alexandrov:2016tnf}
S.~Alexandrov, S.~Banerjee, J.~Manschot, and B.~Pioline, ``{Multiple
  D3-instantons and mock modular forms I},'' {\em Commun. Math. Phys.} {\bf
  353} (2017), no.~1, 379--411,
\href{http://www.arXiv.org/abs/1605.05945}{{\tt 1605.05945}}.

\bibitem{Alexandrov:2017qhn}
S.~Alexandrov, S.~Banerjee, J.~Manschot, and B.~Pioline, ``{Multiple
  D3-instantons and mock modular forms II},'' {\em Commun. Math. Phys.} {\bf
  359} (2018), no.~1, 297--346,
\href{http://www.arXiv.org/abs/1702.05497}{{\tt 1702.05497}}.

\bibitem{Alexandrov:2018lgp}
S.~Alexandrov and B.~Pioline, ``{Black holes and higher depth mock modular
  forms},'' {\em Commun. Math. Phys.} {\bf 374} (2019), no.~2, 549--625,
\href{http://www.arXiv.org/abs/1808.08479}{{\tt 1808.08479}}.

\bibitem{Manschot:2010xp}
J.~Manschot, ``{Wall-crossing of D4-branes using flow trees},'' {\em Adv.
  Theor. Math. Phys.} {\bf 15} (2011), no.~1, 1--42,
  \href{http://www.arXiv.org/abs/1003.1570}{{\tt 1003.1570}}.

\bibitem{Alim:2010cf}
M.~Alim, B.~Haghighat, M.~Hecht, A.~Klemm, M.~Rauch, and T.~Wotschke,
  ``{Wall-crossing holomorphic anomaly and mock modularity of multiple
  M5-branes},'' {\em Commun. Math. Phys.} {\bf 339} (2015), no.~3, 773--814,
  \href{http://www.arXiv.org/abs/1012.1608}{{\tt 1012.1608}}.

\bibitem{Dabholkar:2012nd}
A.~Dabholkar, S.~Murthy, and D.~Zagier, ``{Quantum Black Holes, Wall Crossing,
  and Mock Modular Forms},'' \href{http://www.arXiv.org/abs/1208.4074}{{\tt
  1208.4074}}.

\bibitem{Cheng:2017dlj}
M.~C.~N. Cheng, J.~F.~R. Duncan, S.~M. Harrison, J.~A. Harvey, S.~Kachru, and
  B.~C. Rayhaun, ``{Attractive Strings and Five-Branes, Skew-Holomorphic Jacobi
  Forms and Moonshine},'' {\em JHEP} {\bf 07} (2018) 130,
  \href{http://www.arXiv.org/abs/1708.07523}{{\tt 1708.07523}}.

\bibitem{Collinucci:2008ht}
A.~Collinucci and T.~Wyder, ``{The elliptic genus from split flows and
  Donaldson-Thomas invariants},'' {\em JHEP} {\bf 05} (2010) 081,
\href{http://www.arXiv.org/abs/0810.4301}{{\tt 0810.4301}}.

\bibitem{VanHerck:2009ww}
W.~Van~Herck and T.~Wyder, ``{Black Hole Meiosis},'' {\em JHEP} {\bf 04} (2010)
  047, \href{http://www.arXiv.org/abs/0909.0508}{{\tt 0909.0508}}.

\bibitem{Alexandrov:2022pgd}
S.~Alexandrov, N.~Gaddam, J.~Manschot, and B.~Pioline, ``{Modular bootstrap for
  D4-D2-D0 indices on compact Calabi-Yau threefolds},'' to appear in {\em Adv. Theo. Math. Phys.},
  \href{http://www.arXiv.org/abs/2204.02207}{{\tt 2204.02207}}.

\bibitem{Alexandrov:2023zjb}
S.~Alexandrov, S.~Feyzbakhsh, A.~Klemm, B.~Pioline, and T.~Schimannek,
  ``{Quantum geometry, stability and modularity},''
  \href{http://www.arXiv.org/abs/2301.08066v2}{{\tt 2301.08066v2}}.

\bibitem{Zagier:1975}
D.~Zagier, ``{Nombres de classes et formes modulaires de poids 3/2},'' {\em C.
  R. Acad. Sc. Paris} {\bf 281} (1975) 883--886.

\bibitem{Toda:2011aa}
Y.~Toda, ``{Bogomolov-Gieseker type inequality and counting invariants},''
  \href{http://www.arXiv.org/abs/1112.3411}{{\tt 1112.3411}}.

\bibitem{Feyzbakhsh:2020wvm}
S.~Feyzbakhsh and R.~P. Thomas, ``{Curve counting and S-duality},'' {\em
  {\'E}pijournal de G{\'e}om{\'e}trie Alg{\'e}brique} {\bf 7} (2020)
  \href{http://www.arXiv.org/abs/2007.03037}{{\tt 2007.03037}}.

\bibitem{Feyzbakhsh:2021rcv}
S.~Feyzbakhsh and R.~P. Thomas, ``{Rank $r$ DT theory from rank $0$},''
  \href{http://www.arXiv.org/abs/2103.02915}{{\tt 2103.02915}}.

\bibitem{Feyzbakhsh:2021nds}
S.~Feyzbakhsh and R.~P. Thomas, ``{Rank $r$ DT theory from rank $1$},''
  \href{http://www.arXiv.org/abs/2108.02828}{{\tt 2108.02828}}.

\bibitem{Feyzbakhsh:2022ydn}
S.~Feyzbakhsh, ``{Explicit formulae for rank zero DT invariants and the OSV
  conjecture},'' \href{http://www.arXiv.org/abs/2203.10617}{{\tt 2203.10617}}.

\bibitem{bayer2011bridgeland}
A.~Bayer, E.~Macr{\`\i}, and Y.~Toda, ``Bridgeland stability conditions on
  threefolds I: Bogomolov-Gieseker type inequalities,'' {\em Journal of
  Algebraic Geometry} 117--163.

\bibitem{li2019stability}
C.~Li, ``{On stability conditions for the quintic threefold},'' {\em
  Inventiones mathematicae} {\bf 218} (2019), no.~1, 301--340,
  \href{http://www.arXiv.org/abs/1810.03434}{{\tt 1810.03434}}.

\bibitem{Halder:2023kza}
I.~Halder and Y.-H. Lin, ``{Blackhole/blackring transition},''
  \href{http://www.arXiv.org/abs/2307.13735}{{\tt 2307.13735}}.

\bibitem{ks}
M.~Kontsevich and Y.~Soibelman, ``{Stability structures, motivic
  Donaldson-Thomas invariants and cluster transformations},''
  \href{http://www.arXiv.org/abs/0811.2435}{{\tt 0811.2435}}.

\bibitem{Joyce:2008pc}
D.~Joyce and Y.~Song, ``{A theory of generalized Donaldson-Thomas
  invariants},'' {\em {Memoirs of the Am. Math. Soc.}} {\bf 217} (2012),
  no.~1020,
\href{http://www.arXiv.org/abs/0810.5645}{{\tt 0810.5645}}.

\bibitem{Denef:2007vg}
F.~Denef and G.~W. Moore, ``{Split states, entropy enigmas, holes and halos},''
  {\em JHEP} {\bf 1111} (2011) 129,
\href{http://www.arXiv.org/abs/hep-th/0702146}{{\tt hep-th/0702146}}.

\bibitem{Shmakova:1996nz}
M.~Shmakova, ``{Calabi-Yau black holes},'' {\em Phys. Rev. D} {\bf 56} (1997)
  540--544, \href{http://www.arXiv.org/abs/hep-th/9612076}{{\tt
  hep-th/9612076}}.

\bibitem{Manschot:2011xc}
J.~Manschot, B.~Pioline, and A.~Sen, ``{A Fixed point formula for the index of
  multi-centered N=2 black holes},'' {\em JHEP} {\bf 05} (2011) 057,
  \href{http://www.arXiv.org/abs/1103.1887}{{\tt 1103.1887}}.

\bibitem{Gaiotto:2005gf}
D.~Gaiotto, A.~Strominger, and X.~Yin, ``{New connections between 4-D and 5-D
  black holes},'' {\em JHEP} {\bf 02} (2006) 024,
  \href{http://www.arXiv.org/abs/hep-th/0503217}{{\tt hep-th/0503217}}.

\bibitem{MR1818182}
R.~P. Thomas, ``{A holomorphic {C}asson invariant for {C}alabi-{Y}au 3-folds,
  and bundles on {$K3$} fibrations},'' {\em J. Differential Geom.} {\bf 54}
  (2000), no.~2, 367--438.

\bibitem{Liu:2022agh}
Z.~Liu and Y.~Ruan, ``{Castelnuovo bound and higher genus Gromov-Witten
  invariants of quintic 3-folds},''
  \href{http://www.arXiv.org/abs/2210.13411}{{\tt 2210.13411}}.

\bibitem{pandharipande2009curve}
R.~Pandharipande and R.~P. Thomas, ``Curve counting via stable pairs in the
  derived category,'' {\em Inventiones mathematicae} {\bf 178} (2009), no.~2,
  407--447.

\bibitem{toda2010curve}
Y.~Toda, ``{Curve counting theories via stable objects I. DT/PT
  correspondence},'' {\em Journal of the American Mathematical Society} {\bf
  23} (2010), no.~4, 1119--1157.

\bibitem{bridgeland2011hall}
T.~Bridgeland, ``Hall algebras and curve-counting invariants,'' {\em Journal of
  the American Mathematical Society} {\bf 24} (2011), no.~4, 969--998.

\bibitem{gw-dt}
D.~Maulik, N.~Nekrasov, A.~Okounkov, and R.~Pandharipande, ``Gromov-{W}itten
  theory and {D}onaldson-{T}homas theory. {I},'' {\em Compos. Math.} {\bf 142}
  (2006), no.~5, 1263--1285.

\bibitem{gw-dt2}
D.~Maulik, N.~Nekrasov, A.~Okounkov, and R.~Pandharipande, ``Gromov-{W}itten
  theory and {D}onaldson-{T}homas theory. {II},'' {\em Compos. Math.} {\bf 142}
  (2006), no.~5, 1286--1304.

\bibitem{Pandharipande:2011jz}
R.~Pandharipande and R.~P. Thomas, ``{13/2 ways of counting curves},'' {\em
  Lond. Math. Soc. Lect. Note Ser.} {\bf 411} (2014) 282--333,
  \href{http://www.arXiv.org/abs/1111.1552}{{\tt 1111.1552}}.

\bibitem{Gopakumar:1998ii}
R.~Gopakumar and C.~Vafa, ``{M theory and topological strings. 1.},''
  \href{http://www.arXiv.org/abs/hep-th/9809187}{{\tt hep-th/9809187}}.

\bibitem{Gopakumar:1998jq}
R.~Gopakumar and C.~Vafa, ``{M theory and topological strings. 2.},''
  \href{http://www.arXiv.org/abs/hep-th/9812127}{{\tt hep-th/9812127}}.

\bibitem{Huang:2006hq}
M.-x. Huang, A.~Klemm, and S.~Quackenbush, {\em {Topological string theory on
  compact Calabi-Yau: Modularity and boundary conditions}}, pp.~1--58.
\newblock Springer, 2008.
\newblock \href{http://www.arXiv.org/abs/hep-th/0612125}{{\tt hep-th/0612125}}.

\bibitem{Alexandrov:2012au}
S.~Alexandrov, J.~Manschot, and B.~Pioline, ``{D3-instantons, Mock Theta Series
  and Twistors},'' {\em JHEP} {\bf 1304} (2013) 002,
\href{http://www.arXiv.org/abs/1207.1109}{{\tt 1207.1109}}.

\bibitem{Alexandrov:2019rth}
S.~Alexandrov, J.~Manschot, and B.~Pioline, ``{S-duality and refined BPS
  indices},'' {\em Commun. Math. Phys.} {\bf 380} (2020), no.~2, 755--810,
  \href{http://www.arXiv.org/abs/1910.03098}{{\tt 1910.03098}}.

\bibitem{Bantay:2007zz}
P.~Bantay and T.~Gannon, ``{Vector-valued modular functions for the modular
  group and the hypergeometric equation},'' {\em Commun. Num. Theor. Phys.}
  {\bf 1} (2007) 651--680, \href{http://www.arXiv.org/abs/0705.2467}{{\tt
  0705.2467}}.

\bibitem{Manschot:2008zb}
J.~Manschot, ``{On the space of elliptic genera},'' {\em Commun. Num. Theor.
  Phys.} {\bf 2} (2008) 803--833,
  \href{http://www.arXiv.org/abs/0805.4333}{{\tt 0805.4333}}.

\bibitem{Bouchard:2018pem}
V.~Bouchard, T.~Creutzig, and A.~Joshi, ``{Hecke Operators on Vector-Valued
  Modular Forms},'' {\em SIGMA} {\bf 15} (2019) 041,
  \href{http://www.arXiv.org/abs/1807.07703}{{\tt 1807.07703}}.

\bibitem{CYdata}
``{One Parameter Calabi-Yau higher genus data},''
  \url{http://www.th.physik.uni-bonn.de/Groups/Klemm/data.php}.

\bibitem{Beaujard:2019pkn}
G.~Beaujard, S.~Mondal, and B.~Pioline, ``{Quiver indices and Abelianization
  from Jeffrey-Kirwan residues},'' {\em JHEP} {\bf 10} (2019) 184,
  \href{http://www.arXiv.org/abs/1907.01354}{{\tt 1907.01354}}.

\bibitem{koseki2022stability}
N.~Koseki, ``Stability conditions on Calabi-Yau double/triple solids,'' in {\em
  Forum of Mathematics, Sigma}, vol.~10, p.~e63, Cambridge University Press.
\newblock 2022.
\newblock \href{http://www.arXiv.org/abs/2007.00044}{{\tt 2007.00044}}.

\bibitem{AB-in-progress}
S.~Alexandrov and K.~Bendriss.
\newblock In progress.

\bibitem{bayer2016space}
A.~Bayer, E.~Macr{\`\i}, and P.~Stellari, ``The space of stability conditions
  on abelian threefolds, and on some Calabi-Yau threefolds,'' {\em Inventiones
  mathematicae} no.~3, 869--933.

\bibitem{feyz:effective-restriction-theorem}
S.~Feyzbakhsh, ``An effective restriction theorem via wall-crossing and
  {M}ercat's conjecture,'' {\em Math. Z.} {\bf 301} (2022), no.~4, 4175--4199,
  \href{http://www.arXiv.org/abs/1608.07825}{{\tt 1608.07825}}.

\bibitem{li2009genus}
J.~Li and A.~Zinger, ``On the genus-one Gromov-Witten invariants of complete
  intersections,'' {\em Journal of Differential Geometry} {\bf 82} (2009),
  no.~3, 641--690, \href{http://www.arXiv.org/abs/0507104}{{\tt 0507104}}.

\end{thebibliography}

\end{document}